\documentclass[final,1p,times,authoryear]{elsarticle}
\usepackage{amsmath}
\usepackage{amsthm}
\usepackage{amssymb}
\usepackage{amsfonts}
\usepackage[ruled,linesnumbered,algo2e,vlined,boxed]{algorithm2e}
\usepackage{enumitem}
\usepackage{xcolor}
\usepackage{wrapfig}
\usepackage{hyperref}
\usepackage{subcaption}
\usepackage{textcomp} 

\newtheorem{theorem}{Theorem}
\newtheorem{lemma}[theorem]{Lemma}
\newtheorem{corollary}[theorem]{Corollary}

\newtheorem{proposition}[theorem]{Proposition}
\newtheorem{definition}[theorem]{Definition}

\newtheorem{remark}[theorem]{Remark}

\usepackage{thmtools,thm-restate}
\usepackage[textsize=footnotesize,backgroundcolor=white]{todonotes}

\usepackage{listings}

\newcommand{\RR}{\mathbb{R}}
\newcommand{\CC}{\mathbb{C}}
\newcommand{\ZZ}{\mathbb{Z}}
\newcommand{\QQ}{\mathbb{Q}}

\newcommand{\OB}{\mathcal{O}_B}
\newcommand{\OO}{\mathcal{O}}
\newcommand{\sOB}{\widetilde{\mathcal{O}}_B}
\newcommand{\sOO}{\widetilde{\mathcal{O}}}

\newcommand{\bsz}{\mathcal{L}\xspace}
\newcommand{\res}{\mathtt{res}\xspace}

\newcommand{\cC}{\mathcal{C}\xspace}
\newcommand{\cB}{\mathcal{B}\xspace}

\newcommand{\tT}{\mathtt{T}\xspace}
\newcommand{\maple}{\textsc{maple}\xspace}
\newcommand{\ptopo}{\texttt{PTOPO}\xspace}

\let\svthefootnote\thefootnote
\newcommand\freefootnote[1]{%
  \let\thefootnote\relax%
  \footnotetext{#1}%
  \let\thefootnote\svthefootnote%
}

\begin{document}

\begin{frontmatter}

\title{\ptopo: Computing the Geometry and the Topology of Parametric Curves}

\author{Christina Katsamaki}
\ead{christina.katsamaki@inria.fr}

\author{Fabrice Rouillier}
\ead{Fabrice.Rouillier@inria.fr}

\author{Elias Tsigaridas}
\ead{elias.tsigaridas@inria.fr}
\address{Inria Paris, IMJ-PRG, Sorbonne Universit\'e and Paris Universit\'e}

\author{Zafeirakis Zafeirakopoulos}
\ead{ zafeirakopoulos@gtu.edu.tr}
\address{Institute of Information Technologies, Gebze Technical University, Turkey}

\begin{abstract}
We consider the problem of computing the topology and describing the
  geometry of a parametric curve in $\RR^n$.  We present an algorithm,
  \ptopo, that constructs an abstract graph that is isotopic to the
  curve in the embedding space.  Our method exploits the benefits of
  the parametric representation and does not resort to
  implicitization.

  Most importantly, we perform all computations in the parameter space
  and not in the implicit space.
  When the parametrization involves polynomials of degree at most $d$
  and maximum bitsize of coefficients $\tau$, then the worst case bit
  complexity of \ptopo is
  $ \sOB(nd^6+nd^5\tau+d^4(n^2+n\tau)+d^3(n^2\tau+ n^3)+n^3d^2\tau)$.
  This bound matches the current record bound $\sOB(d^6+d^5\tau)$ for
  the problem of computing the topology of a plane algebraic curve
  given in implicit form.
  For plane and space curves, if
  $N = \max\{d, \tau \}$, the complexity of \ptopo becomes
  $\sOB(N^6)$, which improves the state-of-the-art result, due to
  Alc\'azar and D\'iaz-Toca [CAGD'10], by a factor of $N^{10}$.
  In the same time complexity, we obtain a graph whose straight-line embedding is isotopic to the curve.
  However, visualizing the curve on top of the abstract graph
  construction, increases the bound to $\sOB(N^7)$.
  For curves of general dimension, we can also distinguish between ordinary and non-ordinary real singularities
  and determine their multiplicities in the same expected complexity of \ptopo by employing the algorithm of Blasco and P\'erez-D\'iaz [CAGD'19].
  We have implemented \ptopo in \textsc{maple} for the case of plane
  and space
  curves.  Our experiments illustrate its practical
  nature.

\end{abstract}

\end{frontmatter}

\section{Introduction}

Parametric curves constitute a classical and important topic in
computational algebra and geometry \citep{Sendra:1999:ARR:2378083.2378093} that
constantly receives attention, e.g.,~
\citep{Sed-improper-param-86,Cox-syz,BLY0-spm-2019,SenWinPer-book-08}.
The motivation behind the continuous interest in  efficient algorithms for computing with
parametric curves emanates, among others reasons, by the frequent presence of
parametric representations in computer modeling and computer aided
geometric design,~e.g.,~\citep{10.1007/978-3-642-11620-9_13}.

We focus on computing the topology of a real parametric curve, that is,
the computation of an abstract graph that is isotopic
\cite[p.~184]{bt-ecgcs-06} to the curve in the embedding space.
We design an algorithm, \ptopo,
that applies directly to rational parametric curves of any dimension 
and is complete, in the sense that there are no assumptions on the input.
We consider different characteristics of the parametrization, like
properness and normality, before computing the singularities and other
interesting points on the curve.  These points are necessary for
representing the geometry of the curve, as well as for producing a
certified visualization of plane and space curves.

\paragraph{Previous work}

A common strategy when dealing with parametric curves is
implicitization.  There has been a lot of research effort,
e.g.,~\cite{SedChen-implicit-95,BLY0-spm-2019}
 and the
references therein, in designing algorithms to compute the implicit
equations describing the curve. However, it is also important to
manipulate parametric curves directly, without converting them to
implicit form. For example, in the parametric form it is easier to visualize the curve
and to find points on it. The advantages of the latter become more significant for curves of high dimension.

The study of the topology of a real parametric curve is a
topic that has not received much attention in the literature, in
contrast to its implicit counterpart
\citep{DDRRS-topo-2018,Kobel-complexity}.  The computation of the topology requires
special treatment, since for instance it is not always easy to choose
a parameter interval such that when we plot the curve over it, we
include all the important topological features (like singular and extreme points)
\citep{ALCAZAR2010483}.  Moreover, while
visualizing the curve using symbolic computational tools, the problem
of missing points and branches may arise
\citep{Carlo-2007,Sendra2002NormalPO}.
\cite{ALCAZAR2010483} study the topology of
real parametric curves without implicitizing. They work directly with
the parametrization and address both plane and space real rational
curves.  Our algorithm to compute the topology is to be juxtaposed to
their work.
We also
refer to \cite{Caravantes2014ComputingTT} and \cite{ALBERTI2008631}
for other approaches based on computations by values and subdivision,
respectively.

To compute  the topology of a curve it is essential to detect its
singularities.  This is an important and well studied problem
\citep{ALCAZAR2010483,RUBIO2009490,Kobel-complexity}
of independent interest.
To identify the singularities, we can first compute the implicit representation and
then apply classical approaches \citep{Walker,Fulton}.
Alternatively, we can compute
the singularities using directly the parametrization.  For instance, there are
necessary and sufficient conditions to identify cusps and
inflection points using determinants,
e.g.,~\citep{LI1997491,MANOCHA19921}.

On computing the singularities of a parametric curve, a line of work related to our approach,
does so by means of a univariate
resultant 
\citep{10.1145/77269.77273, Essen, park02, PEREZDIAZ2007835, GUTIERREZ2002545}.
 We can use the Taylor resultant \citep{10.1145/77269.77273} and the $D$-resultant
\citep{Essen} of two polynomials in $K[t]$, to find
singularities of plane curves parametrized by polynomials, where $K$ is
a field of characteristic zero in the first case and of arbitrary characteristic in the
latter, without resorting to the implicit form.
\cite{park02} extends previous results to curves parametrized by
polynomials in affine $n$-space.  The generalization of the
$D$-resultant for a pair of rational functions and its application to
the study of rational plane curves, is due to \cite{GUTIERREZ2002545}.
In \citep{PEREZDIAZ2007835,BlaPer-rspc-19}  they present a method for
computing the singularities of plane curves using a univariate
resultant and characterizing the singularities using its factorization.
Notably \cite{RUBIO2009490}  work on rational parametric curves
in affine $n$-space;
they use generalized resultants to find the parameters of
the singular points.  Moreover, they characterize the singularities
and compute their multiplicities.

\cite{Cox-syz} use the syzygies of the ideal generated by the
polynomials that give the parametrization to compute the
singularities and their structure.  There are state-of-the-art
approaches that exploit this idea and relate the problem of computing
the singularities with the notion of the $\mu$-basis of the
parametrization, e.g.,
\citep{JIA20182}~and~references therein.
\cite{CHIONH200121}  reveal the connection between the
implicitization B\'ezout matrix and the singularities of a parametric
curve.
\cite{BusDAd-sing-plane-12} present  a complete factorization of
the invariant factors of resultant matrices built from birational
parametrizations of rational plane curves in terms of the singular
points of the curve and their multiplicity graph.
Let us also mention the important work on matrix methods \citep{buse:hal-00847802,buse:inria-00468964} for representing the
implicit form of parametric curves, that is suitable for numerical
computations.
\cite{bernardi}  use the projection from
the rational normal curve to the curve and exploit secant varieties.

\paragraph{Overview of our approach and our contributions}

We introduce \ptopo, a complete, exact, and efficient algorithm
(Alg.~\ref{alg:topol}) for computing the geometric properties and the topology of rational parametric
curves in $\mathbb{R}^n$. Unlike~other algorithms, e.g.~\cite{ALCAZAR2010483}, it makes no assumptions on the input curves, such as the
absence of axis-parallel asymptotes, and is applicable to any dimension.
Nevertheless, it does not identify knots for space curves nor it can be used for determining the equivalence of two knots.

If the (proper) parametrization of the curve consists of polynomials of degree
$d$ and bitsize $\tau$, then \ptopo outputs a graph isotopic  \cite[p.184]{bt-ecgcs-06}, \cite{ALCAZAR2020101830} to the
curve in the embedding space, by performing
\[
  \sOB(nd^6+nd^5\tau+d^4(n^2+n\tau)+d^3(n^2\tau+ n^3)+n^3d^2\tau)
\]
bit operations in the worst case
(Thm.~\ref{thm:ptopo-complexity-noBB}), assuming no singularities at infinity. We also provide a Las Vegas
variant with expected complexity
\[
  \sOB(d^6+d^5(n+\tau)+d^4(n^2+n\tau)+d^3(n^2\tau + n^3)+n^3d^2\tau).
\]
If $n = \OO(1)$, the bounds become $\sOB(N^6)$, where
$N = \max\{d, \tau\}$.
The vertices of the output graph correspond to  special points
 on the curve, in whose neighborhood the
 topology is not trivial, given by their parameter values.
Each edge of the graph is associated with
 two parameter values and corresponds to a unique smooth parametric arc. For an embedding isotopic to
the curve, we map every edge of the abstract graph to the corresponding
parametric~arc.

For plane and space curves, our bound improves the previously known
bound due to \cite{ALCAZAR2010483} by a factor
of $\sOB(N^{10})$. The latter
algorithm \citep{ALCAZAR2010483} performs some computations in the implicit space.
On the contrary, \ptopo $ $
  is a fundamentally different approach since we work exclusively in the parameter space and we do not use a
sweep-line algorithm to construct the isotopic graph.
We handle only the parameters that give important points on the curve,
and thus, we avoid performing operations such as univariate root isolation in
an extension field or evaluation of a polynomial at an algebraic number.


Computing singular points is an essential part of \ptopo
(Lem.~\ref{lem:compl_points}).  We chose not to exploit recent
methods, e.g.,~\citep{BlaPer-rspc-19},
for this
task because this would require introducing new variables (to employ the T-resultant).
We employ older techniques,
e.g.,~\citep{RUBIO2009490,PEREZDIAZ2007835,ALCAZAR2010483}, that rely
on a bivariate polynomial system, Eq.~(\ref{system}). 
We take advantage of this system's symmetry and of nearly optimal
algorithms for bivariate system
solving and for computations with real algebraic numbers
\citep{DDRRS-topo-2018,bouzidi:hal-01342211,det-jsc-2009,PAN2017138}.
In particular, we introduce an algorithm
for isolating the roots of over-determined bivariate polynomial systems
by exploiting the Rational
Univariate Representation (RUR)
\citep{blpr-jsc-2015,bouzidi:hal-01342211,Bouzidi:2013:RUR:2465506.2465519}
that has  worst case and expected bit complexity that matches the bounds
for square systems (Thm.~\ref{thm-biv}).
These are key steps for
obtaining the complexity bounds of Thm.~\ref{thm:ptopo-complexity-box}
and Thm.~\ref{thm:ptopo-complexity-noBB}.

 Moreover, our bound matches the current state-of-the-art
complexity bound, $\sOB(d^6 + d^5\tau)$ or $\sOB(N^6)$, for computing
the topology of implicit plane
curves~\cite{DDRRS-topo-2018,Kobel-complexity}.  However, 
if we want to visualize the graph in 2D or 3D, then we
have to compute a characteristic box (Lem.~\ref{lem:BB}) that contains
all the the topological features of the curve and the intersections of
the curve with its boundary.  In this case, the complexity  of \ptopo
becomes $\sOB(N^7)$ (Thm.~\ref{thm:ptopo-complexity-box}).

A preprocessing step of \ptopo consists of
finding a proper reparametrization of the curve (if it is not
proper).  We present explicit bit complexity bounds
(Lem.~\ref{lem:make-proper}) for the algorithm of 
\cite{Perez-proper-06} to compute a proper parametrization.  Another
preprocessing step is to ensure that there are no singularities at
infinity. Lem.~\ref{lem:normal} handles this task and provides
explicit complexity estimates.

Additionally, we consider the case where the embedding of
the abstract graph has straight line edges and not parametric arcs; in
particular for plane curves, we show that the straight
line embedding of the abstract graph in $\RR^2$ is already isotopic to the curve  (Cor.~\ref{cor:embed2}).
For space curves, the procedure supported by
Thm.~\ref{thm:ptopo-complexity-embedding3} adds a few extra vertices to the
abstract graph, so that the straight line embedding  in $\RR^3$ is isotopic to the curve. The
extra number of vertices serves in resolving situations where self-crossings
occur when continuously deforming the graph to the curve. In 
Thm.~\ref{thm:multiplicities} we also prove that for curves of any dimension, we can compute
the multiplicities and the characterization of singular points in the same
bit-complexity as computing the points (in a Las Vegas setting). For that, we use the
method by \cite{BlaPer-rspc-19}, which does not require any further
computations apart from solving the system that gives the parameters of the
singular points (cf.~\cite{ALCAZAR2010483}).

Last but not least, we provide a certified
implementation\footnote{\url{https://gitlab.inria.fr/ckatsama/ptopo}}
of \ptopo in \maple. The implementation computes the topology of plane and space curves and visualizes them.   If the input consists of rational polynomials, our algorithm and the implementation is certified, since, first of all, the algorithm always outputs the
  correct geometric and topological result. This is because we perform exact
  computations with real algebraic numbers based on arithmetic over the
  rationals. Moreover, no assumption that cannot be verified (for example by another algorithm) is made on the input.

\vspace{10pt}

A preliminary version of our work appeared in \cite{krtz-ptopo-20}. Compared
with this version, we add all the missing proofs for the algebraic tools that we use
in our algorithm,
we present the isotopic embedding for plane and space curves (Sec.~\ref{Sembed}),
and analyze the complexity of the algorithm of  \cite{BlaPer-rspc-19} to determine the multiplicities and the character of real singular points (Sec.~\ref{Smult}).

\paragraph{Organization of the paper}
The next section presents our notation and some useful results needed for our proofs.
 In Sec.~\ref{Scurves} we
give the basic background on rational curves in affine $n$-space.  We
characterize the parametrization by means of injectivity and
surjectivity and describe a reparametrization algorithm.  In
Sec.~\ref{sec:special-points} we present the algorithm to compute the
singular, extreme points (in the coordinate directions), and isolated points on the curve.  In
Sec.~\ref{sec:topology} we describe our main algorithm, \ptopo, that
constructs a graph isotopic to the curve in the embedding space
and
its complexity. In Sec.~\ref{Sembed} we expatiate on the isotopic embedding for plane and space curves. 
In Sec.~\ref{Smult}, we study the multiplicities and the character of real singular points for curves of arbitrary dimension.
 Finally, in Sec.~\ref{sec:implementation} we give examples
and experimental results.

\section{Notation and Algebraic Tools}

For a polynomial $f \in \ZZ[X]$, its infinity norm
 is equal to the maximum absolute value of its
coefficients.  We denote by $\bsz(f)$ the logarithm of its infinity
norm.  We also call the latter the bitsize of the polynomial. 
A univariate polynomial is of size $(d,\tau)$ when its degree is at
most $d$ and has bitsize $\tau$.  The bitsize of a rational function
 is the maximum of the bitsizes of the numerator and the
denominator.
We represent an algebraic number $\alpha \in \CC$ by the
\textit{isolating interval representation}.
When $\alpha \in \RR$ (resp. $\mathbb{C}$), it includes a square-free polynomial which vanishes at $\alpha$
and a rational interval  (resp. Cartesian products of rational intervals) containing $\alpha$ and no
other root of this polynomial (see for example \citep{10.5555/328435}).
We denote by $\OO$ (resp. $\OB$) the arithmetic (resp.  bit)
complexity and we use $\sOO$ (resp. $\sOB$) to ignore
(poly-)logarithmic factors.
We denote the resultant of the polynomials $f,g$ with
respect to $x$ by $\res_x(f,g)$ .  For $t\in \CC$, we denote by $\bar{t}$ its complex
conjugate.
We use $[n]$ to signify the set $\{1, \dots, n\}$.


We now present some useful results, needed for our analysis. 

\begin{lemma}
\label{lem:bounds}
Let $A=\sum_{i=0}^m a_iX^i , B=\sum_{i=0}^nb_{i}X^i\in \ZZ[X]$ of degrees $m$ and $n$ and of bitsizes $\tau$ and $\sigma$ respectively. Let $\alpha_1, \dots, \alpha_m$ be the complex roots of $A$, counting multiplicities. Then, for any $\kappa = 1, \dots, m$ it holds that
\[
  2^{-m\sigma-n \tau - (m+n) \log(m+n)}
  < | B(\alpha_{\kappa}) |
  < 2^{m\sigma+n \tau + (m+n) \log(m+n)} .
\]
\end{lemma}

\begin{proof}
  Following \cite{STRZEBONSKI201931}, let
  $R=\res_X(A(X), Y-B(X))\in \ZZ[Y]$. Using the Poisson's formula for
  the resultant we can write
  $R(Y)=a_m^n\prod_{\kappa=1}^m(Y-B(\alpha_{\kappa}))$.  The maximum
  bitsize of the coefficients of $R(Y)$ is at most
  $m\sigma+n \tau + (m+n) \log(m+n)$. We observe that the roots of
  $R(Y)$ are $B(\alpha_\kappa)$ for $\kappa=1,\dots,m$. Therefore,
  using Cauchy's bound we deduce that
  \[
    2^{-m\sigma-n \tau - (m+n) \log(m+n)}< |B(\alpha_{\kappa})| <2^{m\sigma+n \tau + (m+n) \log(m+n)} .
  \]
\end{proof}

Lemmata ~\ref{lem:many-Ugcd} and \ref{lem:many-Bgcd} restate known results on the gcd computation of various univariate and bivariate polynomials.

\begin{lemma}\label{lem:many-Ugcd}
Let $f_1(X),\dots, f_n(X) \in \ZZ[X]$ of sizes $(d,\tau)$. We can compute their $\gcd$, which is of size $(d, \sOO(d+\tau))$, in worst case complexity $\sOB(n(d^3 + d^2\tau))$,  
 with a Monte Carlo algorithm in  $\sOB(d^2 + d \tau)$, 
  or with a Las Vegas algorithm in $\sOB(n(d^2 + d \tau))$. 
\end{lemma}

\begin{proof}
These are known results \citep{gg-mca-13}. We repeat the arguments adapted to our notation.

  \textit{Worst case:} We compute $g$ by performing $n$ consecutive
  $\gcd$ computations, that is
  $$\gcd( f_1, \gcd(f_2, \gcd( \cdots, \gcd(f_{n-1}, f_n))).$$ Since
  each $\gcd$ computation costs $\sOB(d^3+d^2 \tau)$
  \cite[Lem.4]{blpr-jsc-2015},
  the result for this case follows.

  \textit{Monte Carlo:} We perform one gcd computation by allowing
  randomization. 
  If we choose
  integers $a_3,\dots,a_n$ independently at random from the set
  $\{1,\dots, Kd\}$, where $K=\OO(1)$, we get that
  $\gcd(f_1,\dots, f_n) = \gcd(f_1, f_2+a_3f_3+\dots+a_nf_n)$ in
  $\ZZ[x]$, with probability $\ge 1/2$ \cite[Thm.~6.46]{gg-mca-13}.
  This actually computes the monic gcd in $\QQ$. To compute
  the gcd in $\ZZ$ we need to multiply with the gcd of the leading coefficients of
  $f_1, f_2+a_3f_3+\dots+a_nf_n$ and then take the primitive part of the resulting polynomial.
 This is sufficient since the leading coefficient of the gcd in $\ZZ[X]$ divides the leading coefficients of
 the two polynomials. Also,
 by \cite[Cor.~6.10]{gg-mca-13} the monic gcd of two polynomials in $\QQ[X]$ is
 equal to their gcd in $\ZZ[X]$ divided by their leading coefficient.
The gcd of the two leading coefficients of $f_1, f_2 + a_3f_3 + \dots + a_nf_n$ is
 an integer of bitsize $\sOO(\tau)$, therefore this does not pollute the total complexity.

  We compute
  $g^*=\gcd(f_1, f_2+a_3f_3+\dots+a_nf_n)$. Notice that the polynomial
  $f_2+a_3f_3+\dots+a_nf_n$ is asymptotically of size
  $(d,\tau)$. So, it takes $\sOB(d ^2+d \tau)$ to find $g^*$,
  using the probabilistic algorithm in \cite{SCHONHAGE1988365}.

  \textit{Las Vegas:} We can reduce the probability of failure in the
  Monte Carlo variant of the $\gcd$ computation to zero, by performing
  $n$ exact divisions.  In particular, we check if $g^*$ divides
  $f_3, \dots, f_n$.  Using \cite[Ex.10.21]{gg-mca-13}, the bit
  complexity of these operations is in total
  $\sOB(n(d^2+d \tau))$.
\end{proof}

\begin{lemma}\label{lem:many-Bgcd}
Let $f_1(X,Y),\dots, f_n(X,Y) \in \ZZ[X,Y]$ of bidegrees $(d,d)$ and $\bsz(f_i)=\tau$. We can compute their $\gcd$, which is of bitsize $\sOO(d+\tau)$, in worst case complexity $\sOB(n(d^5 + d^4\tau))$, 
   with a Monte Carlo algorithm in  $\sOB(d^3 + d^2 \tau)$, 
  or with a Las Vegas algorithm in $\sOB(n(d^3 + d^2 \tau))$.
\end{lemma}

\begin{proof}
The
  straightforward approach is to perform $n$ consecutive $\gcd$
  computations, that is

  $$\gcd( f_1, \gcd(f_2, \gcd( \cdots, \gcd(f_{n-1}, f_n))).$$  To
  accelerate the practical complexity we sort $f_i$ in
  increasing order with respect to their degree.  Each $\gcd$
  computation costs $\sOB(d^5+d^4 \tau)$
  \cite[Lem.~5]{bouzidi:hal-01342211}, so the total worst case cost is
  $\sOB(n d^5+n d^4 \tau)$.

  Alternatively, we consider the operation $\gcd(f_1,
  \sum_{k=2}^{n}{a_k f_k})$, where
  $a_k$ are random integers, following \cite[Thm.~6.46]{gg-mca-13}.
  The expected cost of this gcd is $\sOB(d^3 + d^2
  \tau)$.  To see this, notice that we can perform a bivariate $\gcd$ in
  expected time
  $\sOO(d^2)$ \cite[Cor.~11.12]{gg-mca-13}, over a finite field with
  enough elements, and the bitsize of the result is $\sOO(d +
  \tau)$ \cite{Mahler-ineq-mpoly-62}.

  Then, for a Las Vegas algorithm, using exact division, we
  test if the resulting polynomial divides all $f_i$, for $2 \leq i
  \leq n$.  This costs
  $\sOB(n(d^3+d^2\tau))$, by adapting \cite[Ex.10.21]{gg-mca-13} to
  the bivariate case.
\end{proof}





\section{Rational curves}
\label{Scurves}

Following  \cite{ALCAZAR2010483} closely,
we introduce  basic notions for rational curves.
Let
$\widetilde{\cC}$ be an algebraic curve over $\CC^n$,
parametrized by the map
\begin{align}
  \phi: & \quad \CC \dashrightarrow \widetilde{\mathcal{C}} \nonumber \\
        &  \quad t \mapsto \big(\phi_1(t), \dots ,\phi_n(t)\big)=
          \Big( \frac{p_{1}(t)}{q_{1}(t)}, \dots, \frac{p_{n}(t)}{q_{n}(t)} \Big) ,
        \label{param}
\end{align}
where $p_{i}, q_{i}\in \ZZ[t]$ are of size $(d, \tau)$ for
$i\in [n]$, and $\widetilde{\mathcal{C}}$ is the
Zariski closure of ${\text{Im}(\phi)}$.
We call $\phi(t)$ a \textit{parametrization} of
$\widetilde{\mathcal{C}}$.

We study the real trace of $\tilde{\mathcal{C}}$, that is
$\mathcal{C} := \tilde{\mathcal{C}} \cap \RR^n$.
A parametrization $\phi$ is chatacterized
 by means of \textit{properness} (Sec.~\ref{sec:proper}) and
\textit{normality} (Sec.~\ref{sec:normal}).
To ensure these properties, one can reparametrize the curve, i.e., apply
a rational change of parameter to the given parametrization. We refer to
\cite[Ch.~6]{SenWinPer-book-08} for more details on~reparametrization.

Without loss of generality, we assume that no coordinate of the parametrization $\phi$
is constant; otherwise we could embed $\tilde{\cC}$ in a lower
dimensional space. We consider that $\phi$ is in \emph{reduced
  form}, i.e., $\gcd(p_{i}(t), q_{i}(t))=1$, for all $i\in [n]$.  The point at
infinity, $\mathbf{p}_\infty$, is the point on $\mathcal{C}$
 we obtain for $t\rightarrow \pm \infty$ (if it exists).
For a parametrization $\phi$, we consider the following system of bivariate
 polynomials:
\begin{align}
  h_i (s,t) = \frac{p_{i}(s)q_{i}(t)-q_{i}(s)p_{i}(t)}{s-t}, \quad \text{ for } i \in [n].
  \label{system}
\end{align}

\begin{remark}
  \label{rem:deriv}
  For every $i\in [n]$  $h_i(s,t)$ is a  polynomial since
  $(s,s)$ is a root of the numerator for every $s$.
  Also,  $h_i(t,t)=\phi'_i(t)q_i^2(t)$ \cite[Lem.~1.7]{DRes}.
\end{remark}

\subsection{Proper parametrization}
\label{sec:proper}
A parametrization is proper if
$\phi(t)$ is injective for almost all points on
$\widetilde{\mathcal{C}}$. In other words, almost every point on
$\widetilde{\mathcal{C}}$ is the image of exactly one parameter value (real or
complex).
For other equivalent definitions of properness, we refer the reader to \cite[Ch.~4]{SenWinPer-book-08}, \citep{RUBIO2009490}.
As stated in \cite[Thm.~1]{ALCAZAR2010483}, a parametrization is proper if and only if
$\deg( \gcd( h_1(s, t), \dots, h_n(s, t))) = 0$. This leads to an
algorithm for checking properness. By applying Lem.~\ref{lem:many-Bgcd} we get the following:

\begin{restatable}{lemma}{checkproper}
  \label{lem:check-proper}
 There is an algorithm that checks if a parametrization $\phi$ is proper
  in  worst-case bit complexity $\sOB(n (d^5+ d^4 \tau))$ and in expected bit complexity
$\sOB(n( d^3+  d^2\tau))$.
\end{restatable}

\begin{proof}
 We construct  the polynomials $h_i(s,t)$ for all $i\in [n]$ in $\OB(nd^2 \tau)$.  Then, we need to
  check if $\deg( \gcd( h_1(s, t), \dots, h_n(s, t))) = 0$
  \cite[Thm.~1]{ALCAZAR2010483}.  For the $\gcd$ computation,
  we employ Lem.~\ref{lem:many-Bgcd} and the result follows.
\end{proof}

If
$\phi$ is a not a proper parametrization, then there always exists a
parametrization $\psi\in \ZZ(t)^n$ and $R(t)\in
\ZZ(t)$ such that $\psi(R(t))= \phi(t)$ and
$\psi$ is proper \cite[Thm.~7.6]{SenWinPer-book-08}.
There are various algorithms for obtaining a proper parametrization,
e.g.,~\citep{Sed-improper-param-86,GUTIERREZ2002545,SenWinPer-book-08,Perez-proper-06,GAO1992459}.
We consider the algorithm in \citep{Perez-proper-06} for its
simplicity; its pseudo-code is in
Alg.~\ref{alg:make-proper}. 

\begin{algorithm2e}[ht]
  \SetKw{RET}{{\sc return}}
  \KwIn{A parametrization
    $\phi\in \ZZ(t)^n$ as in
    Eq.~\eqref{param}}

  \KwOut{A proper parametrization $\psi=(\psi_1,\dots, \psi_n) \in \ZZ(t)^n$}

  \BlankLine

  \lFor{$i \in [n]$ }{
     $H_i(s, t) \gets p_{i}(s) q_{i}(t) - p_{i}(t) q_{i}(s)
     \in \ZZ[s,t]$
  }

  $H \gets \gcd(H_1, \dots, H_n) = C_m(t) s^{m} + \cdots + C_0(t) \in (\ZZ[t])[s]$

  \lIf{ $m = 1$} {\RET $\phi(t)$  }

  Find $k, l \in [m]$ such that:\newline
  $\deg(\gcd(C_k(t), C_l(t))) = 0$ and $ \frac{C_k(t)}{C_l(t)}\not \in \QQ$

  $R(t) \gets \frac{C_k(t)}{C_l(t)}$

  $r \gets \deg(R) = \max\{ \deg(C_k), \deg(C_l)\}$

  $G \gets s\, C_l(t) - C_k(t)$

  \BlankLine

  \For{$i\in [n]$}{
    $F_i \gets x q_i(t) - p_i(t)$

    $L_i(s, x) \gets \res_t(F_i(t, x), G(t, s) )
    = (\tilde{q}_i (s) x - \tilde{p}_i(s))^r$ }

  \RET  $\psi(t) =
   \big( \frac{\tilde{p}_1(t)}{\tilde{q}_1(t)},\dots,\frac{\tilde{p}_n(t)}{\tilde{q}_n(t)}\big)  $
  \caption{\FuncSty{Make\_Proper}($\phi$)}
  \label{alg:make-proper}
\end{algorithm2e}

\begin{restatable}{lemma}{makeproper}
  \label{lem:make-proper}
  Consider a non-proper parametrization of a curve $\mathcal{C}$,
  consisting of
  univariate polynomials of size $(d, \tau)$.
 Alg.~\ref{alg:make-proper}
  computes a proper parametrization of $\mathcal{C}$,
  involving
  polynomials of degree at most $d$ and bitsize $\OO(d^2 + d\tau)$,
  in  $\sOB(n(d^5 + d^4\tau))$, in the worst case.
\end{restatable}

\begin{proof}
The correctness of the algorithm is proved in \citep{Perez-proper-06}. We analyze its complexity.
  The algorithm first computes the bivariate polynomials
   $H_i(s,t) = p_{i}(s) q_{i}(t) - p_{i}(t) q_{i}(s)$ for $i=1,\dots,n$.  They have bi-degree at most $(d, d)$ and bitsize at most
  $2\tau +1$.
  Then, we compute their $\gcd$, which we denote by $H$, in
  $\sOB(n(d^5+d^4 \tau))$ (Lem.~\ref{lem:many-Bgcd}).
  By \citep{Mahler-ineq-mpoly-62} and \cite[Prop.~10.12]{BPR03} we have that $\bsz(H) = \OO(d +\tau)$.
If we write  $H=C_m(t) s^{m} + \cdots + C_0(t)$, it also holds that $\bsz(C_j) = \OO(d+\tau)$, $j=1, \dots, m$.

   If the degree of $H$ is one, then the parametrization is already
   proper and we have nothing to do.
   Otherwise, we
   consider $H$ as a univariate polynomial in $s$ and we find two of
   its coefficients that are relatively prime, using exact division.
   The complexity of this operation is
   $m^2 \times \sOB(d^2 + d \tau) = \sOB(d^4 + d^3\tau)$
   \cite[Ex.~10.21]{gg-mca-13}.

   Subsequently, we perform $n$ resultant computations to get
   $L_1, \dots L_n$, as defined in  Alg.~\ref{alg:make-proper}. From these we obtain the rational functions of
   the new parametrization. We focus on the computation of $L_1$. The
   same arguments hold for all $L_i$.  The bi-degree of $L_1(s, x)$ is
   $(d, d)$ \cite[Prop.~8.49]{BPR03} and
   $\bsz(L_1) = \OO(d^2 + d\tau)$ \cite[Prop.~8.50]{BPR03};
   the latter dictates the bitsize of the new parametrization.

  To compute $L_1$, we consider $F_1$ and $G$ as univariate
  polynomials in $t$ and we apply a fast algorithm for computing the
  univariate resultant based on subresultants
  \citep{LicRoy-sub-res-01}; it performs $\sOO(d)$ operations.  Each
  operation consists of multiplying bivariate polynomials of bi-degree
  $(d, d)$ 
  and bitsize $\OO(d^2 + d\tau)$
  so it costs $\sOB(d^4 + d^3\tau)$.
  We compute the resultant in $\sOB(d^5 + d^4\tau)$.
  We multiply the latter bound by $n$ to conclude the proof.
\end{proof}

\subsection{Normal parametrization}
\label{sec:normal}

Normality of the parametrization concerns the surjectivity of the map
$\phi$. 
The parametrization $\phi(t)$ is $\RR$-normal if for all points
$\mathbf{p}$ on $\mathcal{C}$ there exists $t_0 \in \RR$ such that
$\phi(t_0) = \mathbf{p}$.  When the parametrization is not
$\RR$-normal, the points that are not in the image of $\phi$ for
$t \in \RR$
are $\mathbf{p}_{\infty}$ (if it exists) and the isolated
points that we obtain for complex values of $t$
\cite[Prop.~4.2]{Carlo-2007}.
An $\RR$-normal reparametrization does not always exist.
We refer to \cite[Sec.~7.3]{SenWinPer-book-08} for further details.
However, if $\mathbf{p}_{\infty}$ exists, then we reparametrize the
curve to avoid possible singularities at infinity.  The point
$\mathbf{p}_\infty$ exists if $\deg(p_i) \leq \deg(q_i)$, for all
$i \in [n]$.

\begin{restatable}{lemma}{makenormal}
  \label{lem:normal}
  If $\mathbf{p}_\infty$ exists, then we can reparametrize the curve using
  a linear rational function to
  ensure that $\mathbf{p}_\infty$ is not a singular point, using a Las
  Vegas algorithm in expected time $\sOB(n(d^2 + d\tau))$.  The new
  parametrization involves polynomials of size $(d, \sOO(d+\tau))$.
\end{restatable}

\begin{proof}
  The point at infinity depends on the parametrization. So, for this
  proof, let us denote the point at infinity of $\phi$ by
  $\mathbf{p}_{\infty}^{\phi}$. This point is obtained for $t\rightarrow \infty$.

  The reparametrization consists of choosing $t_0\in \RR$ and applying
  the map $r : t \mapsto \frac{t_0 \, t + 1}{t - t_0}$ to $\phi$, to
  obtain a new parametrization, $\psi = \phi \circ r$.  The point at
  infinity of the new parametrization is
  $\mathbf{p}_{\infty}^{\psi} = \phi(t_0)$.
  We need to ensure that $\mathbf{p}_{\infty}^{\psi} = \phi(t_0)$ is
  not singular. There are $\le d^2$ singular points, so we choose
  $t_0$ uniformly at random from the set $\{1,\dots, Kd^2\}$ where
  $K\ge 2$.  Then, with probability $\ge 1/2$, $\phi(t_0)$ is not
  singular and $\mathbf{p}_{\infty}^{\psi}$ is also not singular.  The
  bound on the possible values of $t_0$ implies that the bitsize of
  $t_0$ is $\OO(\lg(d))$.

  We compute the new parametrization, $\psi$, in
  $\sOB(n(d^2 + d\tau))$ using multipoint evaluation and
  interpolation, by exploiting the fact that the polynomials in $\psi$
  have degrees at most $d$ and bitsize $\sOO(d+\tau)$.

  For a Las Vegas algorithm we need to check if $\phi(t_0)$ is a cusp
  or a multiple point.  For the former, we evaluate $\phi'$ at $t_0$
  (see Rem.~\ref{rem:deriv}).  This costs $\sOB(n d\tau)$
  \cite[Lem.~3]{Bouzidi:2013:RUR:2465506.2465519}.  For the latter, we
  check if
  $\deg( \gcd(\phi_1(t_0) q_1(t) - p_1(t), \dots, \phi_1(t_0) q_1(t) -
  p_1(t))) = 0$ in $\sOB(n(d^2+d\tau))$ (Lem.~\ref{lem:many-Ugcd}).  If
  $\phi'(t_0)$ is not the zero vector and the degree of the $\gcd$ is
  zero, then $\phi({t_0})$ is not singular.
\end{proof}

\begin{remark}
  Since the reparametrizing function in 
   the previous
  lemma is linear, it does not affect properness
  \cite[Thm.~6.3]{SenWinPer-book-08}.
\end{remark}

\section{Special points on the curve}
\label{sec:special-points}

We consider a parametrization $\phi$ of $\mathcal{C}$ as in
Eq.~(\ref{param}), such that $\phi$ is proper and there are no
singularities at infinity.  We highlight the necessity of these
assumptions when needed.  We detect the \textit{parameters} that
generate the \textit{special points} of $\mathcal{C}$, namely the
singular, the isolated, and the extreme points (in the coordinate directions).
We identify the values of the parameter for which $\phi$ is not defined, namely the poles (see Def.~\ref{poles}). In presence
of poles, $\cC$ consists of multiple  components.

\begin{definition}\label{poles}
  The parameters for which $\phi(t)$ is not defined are the
  \textit{poles} of $\phi$. The sets of poles over the complex and the
  reals are:
\[
  \mathtt{T}_P^\CC=\{ t \in \CC : \prod_{i\in[n]}q_i(t)=0 \}
  \, \text{ and } \, \mathtt{T}_P^{\RR}=\mathtt{T}_P^\CC \cap \RR \text{, respectively.}
\]
\end{definition}

We consider the solution set $S$ of the system of Eq.(\ref{system}) over $\CC^2$:

\[
	S=\{(t,s)\in \CC^2 : h_i(t,s)=0 \text{ for all } i \in [n] \} .
\]

\begin{remark}\label{rem:red}
Notice that when $\phi$ is in reduced form, if $(s,t)\in S$ and $(s,t) \in (\CC\setminus \mathtt{T}_P^\CC) \times \CC$, then also $t \not \in \mathtt{T}_P^\CC$
\cite[(in the proof of) Lem.~9]{RUBIO2009490}.
\end{remark}

Next, we present some well-known results
\citep{RUBIO2009490,SenWinPer-book-08} that we adapt to our
notation. 

\paragraph*{\textbf{Singular points}}

Quoting \cite{MANOCHA19921},
``Algebraically, singular points are points on the curve, in whose neighborhood the curve cannot be
represented as an one-to-one and $C^\infty$ bijective map with an open
interval on the real line". Geometrically,
singularities correspond to shape features that are known as cusps and
self-intersections of smooth branches. \textit{Cusps} are points on
the curve where the tangent vector is the zero vector. This is a
necessary and sufficient condition when the parametrization is proper
\citep{MANOCHA19921}.  Self-intersections are \textit{multiple} points,
i.e., points on $\mathcal{C}$ with more than one preimages.


\begin{restatable}{lemma}{singular}
\label{singular}
The set of parameters corresponding to real cusps is
\[
  \mathtt{T}_C = \left\{ t\in \RR\setminus \mathtt{T}_P^{\RR} : (t,t) \in S\right\}.
\]
The set of parameters corresponding to real \textit{multiple points} is
\[
	\mathtt{T}_M=\{t \in \RR \setminus \mathtt{T}_P^\RR : \exists s\neq t, s\in \RR \text{ such that } (t,s) \in S \}.
\]
\end{restatable}

\begin{proof}
The description of $\mathtt{T}_C$ is an immediate consequence of Rem.~\ref{rem:deriv}. It states that $h_i(t,t)=\phi_i'(t) q_i^2(t)$, for $i\in [n]$.


Now let $\mathbf{p} =\phi(t)$ be a multiple point on $\mathcal{C}$. Then, there is $s\in \RR \setminus \mathtt{T}_P^\RR$ with $\phi(t) = \phi(s) \Rightarrow h_i(t,s)=0$
for all $i\in [n]$ and so $t\in \mathtt{T}_M$. Conversely, let $t\in \mathtt{T}_M$ and $s\neq t$, $s\in \RR$ such that $h_i(t,s) =0$ for all $i\in [n]$. From \cite[(in the proof of) Lem.~9]{RUBIO2009490}, when $\phi$ is in reduced form, if $(t,s)\in S$ and $(t,s) \in (\RR\setminus \mathtt{T}_P^\RR) \times \RR$, then also $s\not \in \mathtt{T}_P^\RR$. So, $h_i(t,s) =0 \Leftrightarrow \frac{p_i(t)}{q_i(t)}= \frac{p_i(s)}{q_i(s)}$ for all $i\in [n]$, and thus $\mathbf{p} =\phi(t)= \phi(s)$ is a real multiple point.
\end{proof}

\noindent
Notice that $\mathtt{T}_C$ and $\mathtt{T}_M$ are not necessarily
disjoint, for we may have both cusps and smooth
branches that intersect  at the same point.

 \paragraph*{\textbf{Isolated points}}
An isolated point on a real curve can only occur for complex values of the parameter.
The point at infinity is not isolated because it is the limit of a sequence of real points.
So, additional care is needed in order to avoid cases where the point at infinity is obtained also for complex values of the parameter.

\begin{restatable}{lemma}{isolated}
\label{isolated}
The set of parameters generating isolated points of $\mathcal{C}$ is
  \begin{align*}
	\mathtt{T}_I  = & \{t\in \CC \setminus  (\RR \cup \mathtt{T}_P^\CC) : (t,\overline{t}) \in S \text{ and } \not \exists s\in \RR \text{ s.t. } (t,s)\in S \text{ and } \phi(t) \neq \lim_{s\rightarrow \infty} \phi(s) \}.
\end{align*}
\end{restatable}

\begin{proof}
  Let $\mathbf{p}=\phi(t)\in \RR^n$ be an isolated point, where
  $t\in \CC\setminus (\RR\cup \mathtt{T}_P^\CC) $. Notice that
  $\mathbf{p}$ is also a multiple point, since it holds that
  $\phi_i(t) = \overline{\phi_i(t)} = \phi_i(\overline{t})$ for
  $i \in [n]$. Thus, $h_i(t,\overline{t})=0$ for all $i\in [n]$ and
  $(t,\overline{t})\in S$. Moreover, since $\mathbf{p}$ is isolated,
  there are no real branches through $\mathbf{p}$ and there does not
  exist $s\in \RR$ such that $\phi(t)=\phi(s)\Rightarrow h_i(t,s)=0$,
  for all $i\in [n]$. So, $t\in \mathtt{T}_I$.

  Conversely, let $(t,\overline{t})\in S$ with
  $t\in \CC \setminus \RR \cup \mathtt{T}_P^\CC$. Since $\phi$ is in
  reduced form, we have that $\overline{t} \not\in P^\CC $\cite[(in
  the proof of) Lem.~9]{RUBIO2009490}, therefore
  $h_i(t, \overline{t}) = 0$, for all $i \in [n]$,
  implies that
  $\phi(t)=\phi(\overline{t}) = \overline{\phi(t)}\in \RR^n$. Since
  there does not exist $s\in \RR$ with $\phi(t)=\phi(s)$, $\mathbf{p}$
  is an isolated point on $\mathcal{C}$.
\end{proof}

\paragraph*{\textbf{Extreme points}}
Consider a vector $\vec{\delta}$ and a point on $\mathcal{C}$ whose
normal vector is parallel to $\vec{\delta}$. If the point is not
singular, then it is an extreme point of $\mathcal{C}$ with respect to
$\vec{\delta}$.  We compute the extreme points with respect to the
direction of each coordinate axis. Rem.~\ref{rem:deriv} leads to the
following lemma:
\begin{restatable}{lemma}{extreme}
  \label{extreme}
The set of parameters generating extreme points in the coordinate directions is
\begin{align*}
	\mathtt{T}_E = \big\{ t\in \RR\setminus \mathtt{T}_P^\RR  : \prod_{i\in[n]}h_i(t,t)=0 \text{ and } t\not \in \mathtt{T}_C\cup \mathtt{T}_M\big\}.
\end{align*}
 \end{restatable}

\subsection{Computation and Complexity}
From Lemmata \ref{singular}, \ref{isolated}, and \ref{extreme}, it
follows that given a proper parametrization $\phi$ without singular
points at infinity, we can easily find the poles and the set of
parameters generating cusps, multiple, extreme, and isolated
points. We do so by solving an over-determined bivariate polynomial
system and univariate polynomial equations. Then, we classify the
parameters that appear in the solutions, by exploiting the fact the
system is symmetric. For sake of completeness, we describe the procedure in
Alg.~\ref{alg:compute-points}.


\begin{algorithm2e}[!h]
  \SetKw{RET}{{\sc return}}
  \KwIn{Proper parametrization
    $\phi \in \ZZ(t)^n$ without singularity at infinity, as in Eq.~(\ref{param})}
  \KwOut{Real poles and parameters that give real cusps, multiple, isolated and extreme points with respect to the
direction the coordinate axes.}

  \BlankLine
\tcc{The subroutines \texttt{SOLVE\_R} and \texttt{SOLVE\_C} return the solution set of a univariate polynomial or a system of
polynomials over the real and complex numbers resp.}

  Compute polynomials $h_1(s,t),\dots, h_n(s,t)$

  $\mathtt{T}_P^\RR\gets \bigcup_{i\in n}$ {\texttt{SOLVE\_R}}$(q_i(t)=0)$

  $\mathtt{T}_P^\CC\gets \bigcup_{i\in [n]}$ {\texttt{SOLVE\_C}}$(q_i(t)=0)$

  $S \gets$ \texttt{SOLVE\_C}$(h_1(s,t)=0, \dots, h_n(s,t)=0)$

  $\mathtt{T}_C, \mathtt{T}_M, \mathtt{T}_I, W \gets \emptyset$

  \For{$(s,t)\in S$}{
    \If{$s=t$ and $s\in \RR \setminus \mathtt{T}_P^\RR$}{

        $\mathtt{T}_C \gets \mathtt{T}_C \cup \{t\} $
    }
  \ElseIf{$s\neq t$}{
  	\If{$s\in \RR \setminus \mathtt{T}_P^\RR$}{
           	\If{$t\in \RR$}{
           		$\mathtt{T}_M \gets \mathtt{T}_M \cup \{t\}$
}
\Else{
		$W\gets W\cup \{t\}$
		}

  }
  \ElseIf{$s=\overline{t}$ and $s\not \in \mathtt{T}_P^\CC$ }{
  	$\mathtt{T}_I \gets \mathtt{T}_I \cup \{t\}$
  }
  }
  }

$\mathtt{T}_I \gets \mathtt{T}_I  \setminus W$

\tcc{Extreme points}
$\mathtt{T}_E\gets \bigcup_{i\in n}$ {\texttt{SOLVE\_R}}$(h_i(t,t)=0)$

$\mathtt{T}_E \gets \mathtt{T}_E \setminus (T_E \cap (\mathtt{T}_C \cup \mathtt{T}_M)) $

  \caption{\FuncSty{Special\_Points}($\phi$)}
  \label{alg:compute-points}
\end{algorithm2e}


To compute the Rational Univariate Representation (RUR) \citep{Rouillier1999SolvingZS} of an overdetermined bivariate system (Thm.~\ref{thm-biv}), we employ
Lem.~\ref{lem:sep-form} and Prop.~\ref{lem:RUR},
 which adapt the techniques used in {\cite{bouzidi:hal-01342211}} to
our setting.

\begin{lemma} \label{lem:sep-form}
  Let $f, g, h_1, \ldots, h_n \in \mathbb{Z} [X,Y]$ with degrees bounded by
  $\delta$ and bitsize of coefficients bounded by $L$. Computing a common
  separating element in the form $X + \alpha Y, \alpha \in \mathbb{Z}$ for the $n$+1
  systems of bivariate polynomial equations $\{f = g = 0\}$, $\{f = h_i = 0\}$, $ i = 1
  \ldots n$ needs $\sOB (n (\delta^6+\delta^5L))$ bit operations in the
  worst case, and $\sOB(n (\delta^5+\delta^4L))$ in the expected case with a Las
  Vegas Algorithm. Moreover, the bitsize of $\alpha$ does not exceed $\log (2
  n\delta^4)$.
\end{lemma}

\begin{proof}
  A straightforward strategy consists of simultaneously running  Algorithm 5
  (worst case) or Algorithm 5' (Las Vegas) from {\cite{bouzidi:hal-01342211}}
  on all the systems. The only modifications needed are that the values of $\alpha$
  to be considered are less than $2 n\delta^4$ (twice a bound on the total number
  of solutions of all the systems) and that the exit test is valid if and only
  if it is valid for all the systems.
\end{proof}

\begin{proposition}\label{lem:RUR}
  Let $f, g \in \mathbb{Z} [X, Y]$ with degrees bounded by $\delta$ and
 coefficients' bitsizes bounded by $L$.
  We can compute a rational parametrization $\{ h (T), X = \frac{h_X
  (T)}{h_1 (T)}, Y = \frac{h_Y (T)}{h_1 (T)} \}$ of $f, g$ with $h, h_1, h_X, h_Y
  \in \mathbb{Z} [T]$ with degrees less than $\delta^2$ and coefficients' bitsizes
  in $\sOO(\delta (L + \delta))$, in $\sOB(\delta^5 (L + \delta))$ bit
  operations in the worst case and  $\sOB(\delta^4 (L + \delta) )$ expected
  bit operations with a Las Vegas Algorithm.
\end{proposition}

\begin{proof}
  Algorithms 6 and 6' from {\cite{bouzidi:hal-01342211}} compute an RUR
  decomposition of $f = g = 0$ in $\sOB(\delta^5 (L + \delta))$ bit
  operations in the worst case and \ $\sOB(\delta^4 (L + \delta) )$
  expected bit operations with a Las Vegas Algorithm respectively. They provide
  $s \leqslant \delta$ parametrizations in the form
  $\{ h_i (T), \frac{h_{i, X} (T)}{h_{i, 1} (T)}, \frac{h_{i, Y}
    (T)}{h_{i, 1} (T)} \}$, where $i=1,\dots,s$, with the following
  properties:
  \begin{itemize}
    \item $\prod_{i = 1}^s h_i$ is a polynomial of degree at most $\delta^2$ with
    coefficients of bitsize $\sOO(\delta L + \delta^2)$.

    \item The degrees of $h_{i, 1} (T), h_{X, 1} (T)$ and $h_{Y, 1} (T)$ are
    less then the degree of $h_i$.

    \item The coefficients' bitsizes of $h_{i, 1} (T), h_{X, 1} (T)$ and $h_{Y,
    1} (T)$ are in $\sOB(\delta L + \delta^2)$.
  \end{itemize}
  Also,
  $$\prod_{i = 1}^s h_i, \frac{\sum_{n = 1}^n h_{j, X} \prod_{i \neq j}
    h_i}{\sum_{n = 1}^n h_{j, 1} \prod_{i \neq j} h_i}, \frac{\sum_{n
      = 1}^n h_{j, Y} \prod_{i \neq j} h_i}{\sum_{n = 1}^n h_{j, 1}
    \prod_{i \neq j} h_i}$$ is a rational parametrization of
  the system $\{f = g = 0 \}$, defined by polynomials of degree less
  than $\delta^2$ with coefficients of bitsizes $\sOO(\delta (L + \delta))$ and
  can be computed from the RUR decomposition performing $\OO(s)$
  multiplications of polynomials of degree at most $\delta^2$ with
  coefficients of bitsize $\sOO(\delta (L + \delta))$, which requires
  $\sOB(\delta^4 (L + \delta))$ bit operations.
\end{proof}

\begin{theorem}\label{thm-biv}
  There exists an algorithm that computes the RUR and the isolating
  boxes of the roots of the system
  $\{ h_1(s,t) = \cdots = h_n(s,t) = 0\}$ with worst-case bit
  complexity $\sOB(n(d^6 +d^5 \tau))$.  There is also a Las Vegas
  variant with expected complexity $\sOB(d^6+nd^5+d^5\tau+nd^4\tau)$.
\end{theorem}

\begin{proof}
  Assume that we know a common separating linear element
  $\ell(s,t)=\ell_0+\ell_1s+\ell_2t$ that separates the roots of the
  $n-1$ systems of bivariate polynomial equations $\{h_1 = h_2 = 0\}$,
  $\{h_1 = h_i = 0\}$, for $3\le i \le n$. We can compute $\ell$ with
  $\sOB (n(d^6 + d^5\tau))$ bit operations in the worst case and with
  $\sOB (n(d^5 + d^4\tau))$ expected bit operations with a Las Vegas
  algorithm (Lem.~\ref{lem:sep-form}).

  We denote an RUR
  for $\{h_1 = h_2 = 0\}$ with respect to $\ell$  by
  $\{r (T), \frac{r_s (T)}{r_I (T)}, \frac{r_t (T)}{r_I (T)} \}$.  In addition, for
  $i = 3 \ldots n$, let
  $\{r_i (T), \frac{r_{i, s} (T)}{r_{i, I} (T)}, \frac{r_{i, t}
    (T)}{r_{i, I} (T)} \}$ be the RUR of $\{h_1 = h_i = 0\}$, also
  with respect to $\ell$.  We compute these representations  for all $i = 3 \ldots n$ with
  $\sOB(n(d^6+d^5\tau))$ bit operations in the worst case, and with
  $\sOB(n(d^5+d^4\tau))$ in expected case with a Las Vegas algorithm
  (Lem.~\ref{lem:RUR}).

   Then,  for the system  $\{h_1= h_2
   =  \ldots = h_n = 0\}$ we can define a rational parametrization
   $\{ \chi (T), \frac{r_s (T)}{r_I (T)}, \frac{r_t (T)}{r_I (T)} \}$,
   where
   {
   $$
     \begin{array}{ll}
       \chi(T) = \gcd (& r(T),  r_3(T), \dots, r_n(T), \\
                       &r_s (T)r_{3,I}(T) - r_{3, s} (T)r_I(T), r_t (T)r_{3,I}(T) - r_{3, t}(T)r_I(T), \\
                 & \quad\quad \vdots \\
                 &  r_s (T)r_{n,I}(T) - r_{n, s} (T)r_I(T), r_t (T)r_{n,I}(T) - r_{n, t}(T)r_I(T)) .
     \end{array}
   $$}

 So to compute such a parametrization, we still need to compute the gcd of $3 n
  - 5$ univariate polynomials of degrees at most $d^2$ and coefficients of
  bitsizes in $\sOO(d \tau)$ which needs $\sOB(n(d^6+d^4\tau))$ bit operations in the worst case.
Isolating the roots of such a parametrization requires $\sOB (d^6+d^5\tau)$ according to Alg.~7 from \ {\cite{bouzidi:hal-01342211}}.
\end{proof}

\begin{remark}[RUR and isolating interval representation]
  \label{rem:both-rep}
  If we use Thm.\ref{thm-biv} to solve the over-determined bivariate
  system of the $h_i$ polynomials of Eq.~(\ref{system}), then we obtain in
  the output an RUR for the roots, which is as follows:
  There is a polynomial $\chi(T) \in \ZZ[T]$ of size
  $(\OO(d^2), \sOO(d^2+d\tau))$ and a mapping:
  \begin{align}\label{eq:sols}
    V(\chi) & \rightarrow V(h_1,\dots,h_n) \nonumber \\
    T  & \mapsto  \big(\frac{r_s(T)}{r_I(T)},  \frac{r_t(T)}{r_I(T)}\big) ,
  \end{align}
  that defines an one-to-one correspondence between the roots of $\chi$
  and those of the system. The polynomials $r_s$, $r_t$, and $r_I$ are in
  $\ZZ[T]$ and have also size  $(\OO(d^2), \sOO(d^2+d\tau))$.
  %

  Taking into account the cost to compute this parametrization of the
  solutions (Thm.\ref{thm-biv}), we can also compute
  the resultant of $\{h_1, h_2\}$ with respect to $s$ or $t$ at no extra cost.  Notice
  that both resultants are the same polynomial, since the system is
  symmetric.  Let $R_s(t) =
  \res_s(h_1,h_2)$.
  It is of size $(\OO(d^2), \OO(d^2+d\tau))$ \cite[Prop.~8.46]{BPR03}.

  Under the same bit complexity, we can sufficiently
  refine the isolating boxes of the solutions of the bivariate system
  (computed in Thm.\ref{thm-biv}), so that every root
  $(\frac{r_s(\xi)}{r_I(\xi)}, \frac{r_t(\xi)}{r_I(\xi)})$, where
  $\chi(\xi)=0$, has a representation as a pair of algebraic numbers
  in isolating interval representation:
  \begin{align}\label{eq:isol}
    ((R_s, I_{1,\xi}\times I_{2,\xi}),  (R_s, J_{1,\xi}\times J_{2,\xi})) .
  \end{align}
 Both coordinates in the latter representation, are algebraic numbers which are roots of the same polynomial.
  Moreover, $I_{2,\xi}, J_{2,\xi}$ are empty sets when the
  corresponding algebraic number is real. Therefore, we can
  immediately distinguish between real and complex parameters.  At the
  same time, we associate to each isolating box of a root of $R_s$ the
  algebraic numbers $\rho=(\chi, I_\rho\times J_\rho)$ for which it
  holds that $\frac{r_s(\rho)}{r_I(\rho)}$ projects inside this
  isolating box.  We can interchange between the two
  representations in constant time and this will simplify our
  computations in the sequel.
\end{remark}

\begin{lemma}\label{lem:compl_points}
Let $\mathcal{C}$ be a curve with a proper parametrization $\phi(t)$ as in Eq.~(\ref{param}), that has no singularities at infinity. We compute the real poles of $\phi$ and the parameters corresponding to singular, extreme (in the coordinate directions), and isolated points of $\mathcal{C}$ in worst-case bit complexity
\[
	\sOB(nd^6+nd^5\tau+d^4(n^2+n\tau)+d^3(n^2\tau+ n^3)+n^3d^2\tau),
\]
and using a Las Vegas algorithm in expected bit complexity
\[
  \sOB(d^6+d^5(n+\tau)+d^4(n^2+n\tau)+d^3(n^2\tau + n^3)+n^3d^2\tau).
\]
\end{lemma}

\begin{proof}
The proof is an immediate consequence of the following:

\noindent $\bullet$ \textit{We compute all $h_i \in \ZZ[s,t]$ in  $\sOB(n d^2\tau)$}:
 To construct each $h_i$ we perform $d^2$ multiplications
  of numbers of bitsize $\tau$; the cost for this  is $\sOB(d^2\tau)$.
  The  bi-degree of each is at most
  $(d, d)$ and $\bsz(h_i) \leq 2\tau + 1 =\OO(\tau)$.

\noindent $\bullet$  \textit{The real poles of $\phi$ are computed in  $\sOB(n^2 (d^4 + d^3 \tau))$: }
To find the poles of
  $\phi$, we isolate the real roots of each polynomial $q_i(t)$, for
  $i \in [n]$. This costs $\sOB(n(d^3 + d^2 \tau))$ \citep{PAN2017138}.  Then we sort the
  roots in $\sOB( n \, d \, n(d^3 + d^2 \tau)) = \sOB(n^2 (d^4 + d^3 \tau))$.

\noindent $\bullet$ \textit{The parameters corresponding to cusps, multiple and isolated points of $\mathcal{C}$  are computed in $\sOB(n(d^6+ d^5 \tau))$:}

\noindent  We solve the bivariate system of Eq.~(\ref{system}) in
  $\sOB(n(d^6+d^5\tau))$ or in expected time
  $\sOB(d^6+nd^5+d^5\tau+nd^4\tau)$ (Thm.~\ref{thm-biv}).  Then we have
  a parametrization of the solutions of the bivariate system
  of Eq.~(\ref{system}) of the form of  Eq.~(\ref{eq:sols}) and in the same time of
  the form of Eq.~(\ref{eq:isol}) (see Rem.~\ref{rem:both-rep}).  Some
  solutions $(s,t)\in S$ may not correspond to points on the curve,
  since $s,t$ can be poles of $\phi$.  Notice that from
  Rem.~\ref{rem:red}, $s$ and $t$ are either both poles or neither of
  them is a pole.  We compute $g_s = \gcd(R_s, Q)$, where
  $Q(t)=\prod_{i\in[n]}q_i(t)$, and the $\gcd$-free part of $R_s$ with
  respect to $Q$.  This is done in
  $\sOB( \max\{n,d\}\, (nd^3\tau+nd^2\tau^2))$~\cite[Lem.~5]{blpr-jsc-2015}.


\noindent Every root of $R^*_s$ is an algebraic number of the form $(R_s, I_{1,\xi}\times I_{2,\xi})$, for some $\xi$ that is root of $\chi$.
We can easily determine if it corresponds to a cusp, a multiple or an isolated point; when real (i.e., $I_{2,\xi}=\emptyset$) it corresponds to a cusp of $\cC$ if and only if $((R_s, I_{1,\xi}),(R_s, I_{1,\xi}))$ is in $S$. Otherwise, it corresponds to a multiple point.
When it is complex (i.e., $I_{2,\xi}\neq\emptyset$),  it corresponds to an isolated pont of $\cC$ if and only if $((R_s, I_{1,\xi}\times I_{2,\xi}),(R_s,  I_{1,\xi}\times (-I_{2,\xi})))\in S$ and there is no root in $S$ of the form $((R_s, I_{1,\xi}\times I_{2,\xi}),(R_s, J_{1,\xi'}\time \emptyset))$.

\noindent $\bullet$ \textit{The parameters corresponding to extreme points of $\mathcal{C}$  with respect to the
direction of each coordinate axis are computed in $\sOB(d^4n\tau+d^3(n^2\tau+n^3)+d^2n^3\tau)$:} 

  \noindent For all $i\in [n]$, $h_i(t,t)$ is a univariate
  polynomial of size $(\OO(d), \OO(\tau))$.
 Then,  $H(t)=\prod_{i\in[n]}h_i(t,t)$ is of size $(\OO(nd),\sOO(n\tau))$. 
  The parameters that correspond to the extreme points are among the roots
  of $H(t)$.  To make sure that poles and parameters that
  give singular points are excluded, we compute $\gcd(H, Q\cdot R_s)$,
  where $Q(t)=\prod_{i\in[n]}q_i(t)$, and the $\gcd$-free part of $H$
  with respect to $Q\cdot R_s$, say $H^*$.  Since $Q\cdot R_s$ is a
  polynomial of size $(d^2+nd, (d+n)\tau)$, the computation of the
  $\gcd$ and the $\gcd$-free part costs
  $\sOB(n(d^4\tau+nd^3\tau+n^2d^2\tau))$ \cite[Lem.~5]{blpr-jsc-2015}.
  Then, $H=\gcd(H, Q\cdot R_s) H^*$, and the real roots of $H^*$ give the parameters that correspond to the extreme
  points. We isolate the real roots of $H^*$ in
  $\sOB(n^3(d^3+d^2\tau))$, since it is a polynomial of size $( \OO(nd),\sOO(n(d+\tau)))$.
\end{proof}

\section{\ptopo: Topology and Complexity}
\label{sec:topology}

We present \ptopo, an algorithm to construct an abstract graph $G$ that
is isotopic  \cite[p.184]{bt-ecgcs-06} to $\mathcal{C}$ when we embed it in $\RR^n$.
We emphasize that, currently, we do not identify knots in
the case of space curves.
 The embedding consists of a graph whose vertices are
points on the curve given by their parameter values.
The edges are smooth parametric arcs that we can continuously deform to branches of $\mathcal{C}$
without any topological changes. We need to specify a
bounding box in $\RR^n$ inside which the constructed graph results in
an isotopic embedding to $\mathcal{C}$.
We comment at the end of the section on the case where an arbitrary
box is provided at the input.
We determine a bounding box in $\RR^n$, which we call
\emph{characteristic}, that captures all the topological information
of $\mathcal{C}$. 
\begin{definition}
  \label{def:char-box}
  A \textit{characteristic box} of
  $\mathcal{C}$ is a box enclosing a subset of
  $\RR^n$ that intersects all components of
  $\mathcal{C}$ and contains all its singular, extreme (in the coordinate directions), and isolated
  points.
\end{definition}


Let $\mathcal{B}_{\mathcal{C}}$ be a characteristic box of
$\mathcal{C}$. If $\cC$ is bounded, then
$\cC \subset \mathcal{B}_{\cC}$.  If $\cC$ is unbounded, then the branches
of $\cC$ that extend to infinity intersect the boundary of
$\mathcal{B}_{\cC}$.  A branch of the curve extends to infinity if for
$t\rightarrow t_0$, it holds $||\phi(t)|| >M$, for any $M>0$, where
$t_0\in \RR \cup\{ \infty\}$.  Lem.~\ref{lem:BB} computes a
characteristic box using the degree and bitsize of
the polynomials in the parametrization of Eq.~(\ref{param}).

\begin{lemma}
  \label{lem:BB}
  Let $\mathcal{C}$ be a curve with a parametrization as in
  Eq.~(\ref{param}).  For
  $b=15d^2(\tau+\log d)=\OO(d^2\tau)$,
 $\mathcal{B}_\mathcal{C}=[-2^{b},2^{b}]^n$ is a~characteristic~box~of~$\mathcal{C}$.
\end{lemma}

\begin{proof}
  We estimate the maximum and minimum values of $\phi_i$, $i\in [n]$,
  when we evaluate it at
  the parameter values that correspond to
 special points and also at each pole
  that is not a root of $q_i$.

Let $t_0$ be a parameter that corresponds to a cusp or an extreme point with respect to the $i$-th direction. Then, it is a root of $\phi_i'(t)$. Let $N(t)=p_i'(t)q_i(t)-p_i(t)q_i'(t)$ the numerator of $\phi_i'(t)$. Then $N(t_0)=0$.
The degree of $N(t)$ is $\le 2d-1$ and $\bsz(N) \le 2^{2\tau+\log d +1}$.
From Lem.~\ref{lem:bounds} we conclude that $|p_i(t_0)|\le 2^{4d\tau+d \log(d)+ (3d-1)\log (3d-1) +d-\tau}$. Analogously, it holds that  $|q_i(t_0)|\ge 2^{-4d\tau-d \log(d)- (3d-1)\log (3d-1) -d+\tau}$. Therefore,
\[
|\phi_i(t_0)|\le 2^{2(4d\tau+d \log(d)+ (3d-1)\log (3d-1) +d-\tau)} .
\]
Now, let $(t_1,t_2)$ be two parameters corresponding to a
multiple point of $\mathcal{C}$, i.e., $(t_1,t_2)$ is a root of the
bivariate system in Eq.~(\ref{system}). Take any $j, k \in [n]$ with
$j\neq k$ and let $R(t)=\res_s(h_j, h_k)$. It holds that
$R(t_1)=0$. The degree of $R$ is $\le 2d^2$ and
$ \bsz(R) \le 2d(\tau+\log(d)+\log(d+1)+1) $ \cite[Prop.~8.29]{BPR03}.
By applying Lem.~\ref{lem:bounds}, we deduce that
\[
|\phi_i(t_1)|\le 2^{4d^2(\tau+\log(d)+\log(d+1)+1)+4d^2\tau+(2d^2+d)\log (2d^2+d)} .
\]
Let $t_3$ be a pole of $\phi$ with $q_j(t_3)= 0$, for some $j\neq i$. If $\phi_i(t_3)$ is defined, applying Lem.~\ref{lem:bounds} gives
\[
    |\phi_i(t_3)|\le 2^{4d\tau +4d\log2d} .
\]
To conclude, we take the maximum of the three bounds. However, to
simplify notation, we slightly overestimate the latter bound.
\end{proof}

The vertices of the embedded graph must include the singular and the
isolated points of $\mathcal{C}$. Additionally, to rigorously
visualize the geometry of $\mathcal{C}$, we consider as vertices
the extreme points of $\mathcal{C}$, with respect to all coordinate
directions, as well as the intersections of $\cC$ with the boundary
of the bounding box.
We label the vertices of $G$ using the corresponding parameter values generating these points, and we connect them accordingly.
Alg.~\ref{alg:topol} presents the pseudo-code of \ptopo and here we give some more details on the various steps.

We construct  $G$ as follows:
First, we compute the poles and the sets
$\mathtt{T}_C, \mathtt{T}_M, \mathtt{T}_E$, and $\mathtt{T}_I$ of
parameters corresponding to ``special points''.
Then, we compute the characteristic box of $\cC$, say $\cB_{\cC}$.   We compute the set $\mathtt{T}_B$ of
parameters corresponding to the intersections of $\mathcal{C}$ with
the boundary of $\cB_\mathcal{C}$ (if any). Lem.~\ref{lem:int} describes this
procedure and its complexity.

\begin{lemma}
  \label{lem:int}
  Let $\mathcal{B}=[l_1,r_1]\times \dots \times [l_n,r_n]$ in $\RR^n$
  and $\bsz(l_i)=\bsz(r_i)=\sigma$, for $i\in[n]$. We can find the
  parameters that give the intersection points of $\phi$ with the
  boundary of $\mathcal{B}$ in $\sOB(n^2d^3+n^2d^2(\tau+\sigma))$.
\end{lemma}

\begin{proof}
  For each $i\in [n]$ the polynomials $q_i(t)l_i-p_i(t)=0$ and
  $q_i(t)r_i-p_i(t)=0$ are of size $(\OO(d), \OO(\tau+\sigma))$. So, we compute isolating intervals for
  all their real solutions in $\sOB(d^2 (\tau+\sigma))$
  \citep{Pan02jsc}.
  For any root $t_0$ of each of these polynomials,
  since $\phi$ is in reduced form (by assumption),
  we have  $t_0 \not \in \tT_P^{\RR}$.
  We check if
  $\phi_j(t_0) \in [l_j, r_j]$, $j\in [n] \setminus i$.  This requires
  3 sign evaluations of univariate polynomials of size
  $(d, \tau+ \sigma)$ at all roots of a polynomial of size
  $(d,\tau+\sigma)$. The bit complexity of performing these operations
  for all the roots is $\sOB(d^3+d^2(\tau+\sigma))$
  \cite[Prop.~6]{STRZEBONSKI201931}. Since we repeat this procedure $n-1$ times for
  every $i\in[n]$, the total cost is
  $\sOB(n^2d^3+n^2d^2(\tau+\sigma))$.
\end{proof}

We partition
$\mathtt{T}_C \cup \mathtt{T}_M \cup \mathtt{T}_E \cup \mathtt{T}_I
\cup \mathtt{T}_{B}$ into groups of parameters that correspond to the
same point on $\cC$.
For each group, we add a vertex to $G$ if and only if the
corresponding point is strictly inside the bounding box $\mathcal{B}$; for the
characteristic box it is strictly inside by construction.

\begin{lemma}
  \label{lem:group-params}
  The graph $G$ has $\kappa =\OO(d^2+nd)$ vertices, which can be computed using $\OO(d^2+nd)$ arithmetic operations.
\end{lemma}

\begin{proof}
Since
$ \mathtt{T}_B\cap \mathtt{T}_M =
\emptyset$ and $ \mathtt{T}_E\cap \mathtt{T}_M =
\emptyset$,
  to each parameter in $\mathtt{T}_B$ and $\mathtt{T}_E$
  corresponds a unique point on $\mathcal{C}$.
  So for every $t\in  \mathtt{T}_B\cup \mathtt{T}_E$ we
  add a vertex to $G$, labeled by the respective parameter. Next, we
  group the parameters in
  $\mathtt{T}_C\cup \mathtt{T}_M \cup \mathtt{T}_I$ that give the same
  point on $\mathcal{C}$ and we add for each group a vertex at $G$
  labeled by the corresponding parameter values.

   Grouping the
  parameters is done as follows:
  For every $t\in \mathtt{T}_C \cup \mathtt{T}_M $ we add a vertex to
  $G$ labeled by the set
  $\left\{ s\in \RR : (s,t)\in
    S\right\}\cup \{t\}$ and for every $t\in \mathtt{T}_I $ we add a
  vertex to $G$ labeled by the set
  $\left\{ s\in \CC : (s,t)\in
    S\right\}\cup \{t\}$. 
    We compute these sets simply by reading the elements of $S$.

 It holds that $\mathtt{T}_B = \OO(nd)$, $\mathtt{T}_E = \OO(nd)$ and $|S| = \OO(d^2)$. Since for each vertex, we can find the parameters that give the same point in constant time, the result follows.
\end{proof}

We denote by $v_1,\dots, v_\kappa$ the vertices (with distinct labels)
of G and by $\lambda(v_1),\dots, \lambda(v_\kappa)$ their label sets
(i.e., the parameters that correspond to each vertex).  Let
$\mathtt{T}$ be the \textit{sorted} list of parameters in
$\mathtt{T}_C\cup \mathtt{T}_M \cup \mathtt{T}_E \cup
\mathtt{T}_B$
(notice that we exclude the parameters of the isolated points).
If for two consecutive elements $t_1<t_2$ in $\mathtt{T}$, there
exists a pole $s \in \mathtt{T}_P^\RR$ such that $t_1 < s < t_2$, then
we split $\mathtt{T}$ into two lists: $\mathtt{T}_1$ containing the
elements $\le t_1$ and $\mathtt{T}_2$ containing the elements
$\ge t_2$. We continue recursively for $\mathtt{T}_1$ and
$\mathtt{T}_2$, until there are no poles between any two elements of
the resulting list. This procedure partitions $\mathtt{T}$ into $\mathtt{T}_1, \dots, \mathtt{T}_{\ell}$.

To add edges to $G$, we consider each $\mathtt{T}_i$ with more than one element,
 where $i \in [\ell]$, independently.  For any
consecutive elements $t_1 < t_2$ in $\mathtt{T}_i$, with
$t_1 \in \lambda(v_{i,1})$ and $t_2 \in \lambda(v_{i,2})$, we add the
edge $\{v_{i,1}, v_{i,2}\}$. To avoid multiple edges, we make the convention that we add an edge between $v_{i,j}$, $j=1,2$, and an (artificial) intermediate point corresponding to a parameter in $(t_1,t_2)$. 
If $\mathbf{p}_\infty$ exists, we add an edge to the graph
connecting the vertices corresponding to the last element of
$\mathtt{T}_{\ell}$ and the first element of the $\mathtt{T}_1$.

\begin{algorithm2e}[!h]
  \SetKw{RET}{{\sc return}}
  \KwIn{ A proper parametrization
    $\phi \in \ZZ(t)^n$ without singular points at infinity. }

  \KwOut{Abstract graph $G$}

  \BlankLine

Compute real poles $ \mathtt{T}_P^\RR$.

Compute parameters of `special points' $\mathtt{T}_C, \mathtt{T}_M, \mathtt{T}_E, \mathtt{T}_I$.

\tcc{Characteristic box}
	$b \gets 15d^2(\tau+\log d), \quad \cB_{\cC} \gets [-2^b,2^b]^n$

$\mathtt{T}_B \gets$ parameters that give to intersections of $\mathcal{C}$ with $\cB_\cC$

Construct the set of vertices of $G$ using Lem.\ref{lem:group-params}

 Sort the list of all the parameters $\mathtt{T}=[\mathtt{T}_C, \mathtt{T}_M, \mathtt{T}_E, \mathtt{T}_B]$.


  Let $T_1,\dots, T_{\ell}$ the sublists of $\mathtt{T}$ when split at parameters in $\mathtt{T}_P^\RR$


  \For{every list $\tT_i=[t_{i,1},\dots, t_{i,k_i} ]$}{
    \For{ $j=1, \dots, k_i -1$}{
        Add the edge $\{t_{i,j}, t_{i,j+1}\}$ to the graph 
    }
  }

\If{ $\mathbf{p}_\infty$ exists}{
	Add the edge  $\{t_{1,1}, t_{\ell,k_\ell}\}$ to the graph
}

  \caption{\FuncSty{PTOPO}($\phi$) $\quad$(Inside the characteristic box)}
  \label{alg:topol}
\end{algorithm2e}

\begin{theorem}[\ptopo inside the characteristic box]
  \label{thm:ptopo-complexity-box}
  Consider a proper parametrization $\phi$ of curve $\cC$ involving
  polynomials of degree $d$ and bitsize $\tau$, as in Eq.(\ref{param}),  that has no singularities at infinity.
  Alg.~\ref{alg:topol} outputs a graph $G$ that, if embedded in $\RR^n$, is isotopic to $\cC$, within the characteristic box
  $\cB_{\cC}$. It has worst case complexity
  \[
  	\sOB(d^6(n+\tau)+nd^5\tau+n^2d^4\tau+d^3(n^2\tau+n^3)+n^3d^2\tau) ,
  \]
  while its expected complexity is
    \[
  	\sOB(d^6\tau+nd^5\tau+n^2d^4\tau+d^3(n^2\tau+n^3)+n^3d^2\tau).
  \]

  If $n = \OO(1)$, then  bounds become $\sOB(N^7)$, where
  $N = \max\{d, \tau\}$.
\end{theorem}

\begin{proof}
  We count on the fact that $\phi$ is continuous in
  $\RR \setminus \tT_P^\RR$. Thus, for each real interval $[s,t]$
  with $[s,t]\cap \tT_P^\RR = \emptyset$, there is a parametric arc
  connecting the points $\phi(s)$ and $\phi(t)$.
Since for any (sorted) list $\tT_i$, for $i \in [\ell]$, the interval defined by the minimum and maximum value of its elements
has empty intersection with $\mathtt{T}_P^\RR$,
 then for any $s, t \in \tT_i$ there
  exists a parametric arc connecting $\phi(s)$ and
  $\phi(t)$ and it is entirely contained in $\mathcal{B}_\cC$.
  If $\mathbf{p}_\infty$ exists, then $\mathbf{p}_\infty$ is inside
  $\cB_{\cC}$. Let $t_{1,1}, t_{\ell,k_\ell}$ be the first element of the first list
   and the last element of the last list.
 There is a parametric arc
  connecting $\phi(t_{1,1})$ with $\mathbf{p}_\infty$  and
  $\mathbf{p_\infty}$ with $\phi(t_{\ell,k_\ell})$. So we add the edge
  $\{t_{1,1}, t_{\ell,\kappa_\ell}\}$ to $G$.  Then,
  every edge of $G$ is embedded to a unique smooth parametric arc
and the embedding of $G$ can
   be trivially continuously deformed~to~$\cC$.

For the complexity analysis, we know from Lem.\ref{lem:compl_points} that steps 1-2 can be performed  in wost-case bit complexity
\[
	\sOB(nd^6+nd^5\tau+d^4(n^2+n\tau)+d^3(n^2\tau+ n^3)+n^3d^2\tau) ,
\]
and in expected bit complexity
\[
	\sOB(d^6+d^5(n+\tau)+d^4(n^2+n\tau)+d^3(n^2\tau + n^3)+n^3d^2\tau) ,
  \]
  using a Las Vegas algorithm. From Lemmata \ref{lem:BB},
  \ref{lem:int}, and \ref{lem:group-params} steps~4-5 cost $\sOB(n^2(d^3\tau))$.

To perform steps 6-7 we must sort all the parameters in $\mathtt{T}\cup \mathtt{T}_P^\RR$, i.e.,
we sort $\OO(d^2+nd)$ algebraic numbers. The parameters that correspond to cusps and extreme points with respect to the $i$-th coordinate direction can be expressed as roots of $\prod_{i\in[n]}h_i(t,t)$, which is of size $(nd,n\tau)$. The poles are roots of $ \prod_{i\in[n]}q_i(t)$, which has size $(nd,n\tau)$. The parameters that correspond to multiple points are roots of $R_s$ which has size $(d^2,d\tau)$. At last, parameters in $\mathtt{T}_B$ are roots of a polynomial of size $(d,d^2\tau)$.


We can consider all these algebraic numbers together as roots of a single univariate polynomial (the product of all the corresponding polynomials).
It has degree $\OO(d^2 + nd)$ and bitsize $\sOO(d^2\tau + n \tau)$.
Hence, its separation bound is $\sOO(d^4\tau + nd^3\tau + n d^2\tau + n^2d\tau)$.
To sort the list of all the algebraic numbers, we have to perform $\OO(d^2 + nd)$ comparisons
and each costs $\sOO(d^4\tau + nd^3\tau + n d^2\tau + n^2d\tau)$.
Thus, the overall cost for sorting is  $\sOB(d^6\tau +nd^5\tau+n^2d^4\tau+n^2d^3\tau+n^3d^2\tau)$.
The overall bit complexities in the worst and expected case follow by summing the previous bounds.
\end{proof}

Following the proof of Thm.~\ref{thm:ptopo-complexity-box} we notice
that the term $d^6\tau$ in the worst case bound is due to the
introduction of the intersection points of $\mathcal{C}$ with
$\mathcal{B}_\cC$.
For visualizing the curve within $\mathcal{B}_\cC$, these points are
essential and we cannot avoid them.  However, if we are interested
only in the topology of $\cC$, i.e., the abstract graph
$G$, these points are not important any more. We sketch a procedure to
avoid them and gain a factor of $d$ in the complexity bound:

Assume that we have not computed the points on
$\mathcal{C}\cap \mathcal{B}_\cC$. We split again the sorted list
$\mathtt{T}=[\mathtt{T}_C, \mathtt{T}_M, \mathtt{T}_E]$ at the real
poles, and we add an artificial parameter at the beginning and at the
end of each sublist.  The rest of the procedure remains unaltered.

To verify the correctness of this approach, it suffices to prove
that the graph that we obtain by this procedure, is
isomorphic to the graph $G$. It is immediate to see that the latter
holds, possibly up to the dissolution of the
vertices corresponding to the first and last artificial vertices.
Adding these artificial parameters does not affect the overall
complexity, since we do not perform any algebraic operations.
Therefore, the bit complexity of the algorithm is determined by the
complexity of computing the parameters of the special points
(Lem.\ref{lem:compl_points}), and so we have the following theorem:

\begin{figure}[t!]
  \begin{center}
    \includegraphics[scale=0.29]{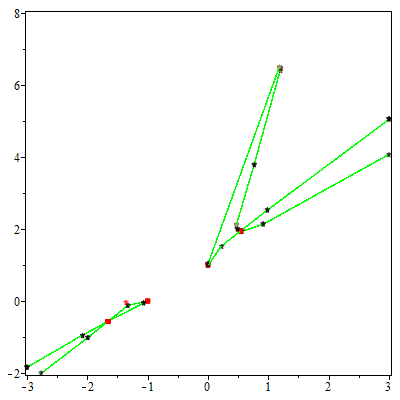}
    \includegraphics[scale=0.29]{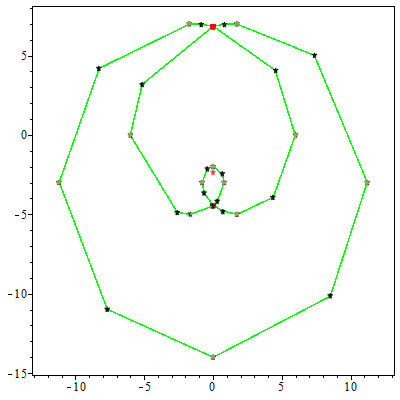}
    \caption{\emph{\textmd{The left figure is the output of \textsc{ptopo} for the  parametric curve
        $(\frac {3\,t^{2}+3\,t+1}{t^{6}-2\,t^{4}-3\,t-1},\frac {
            ( t^{4}-2\,t+2 ) t^{2}}{t^{6}-2\,t^{4}-3\,t-1} )$,
        while the right figure is the output for the curve
        $\small (\tfrac {6\,t^{8}-756\,t^{6}+3456\,t^{5}-31104\,t^{3}+61236\,{
              t}^{2}-39366}{t^{8}+36\,t^{6}+486\,t^{4}+2916\,t^{2}+6561}$,
        $\tfrac { -18( 6\,t^{6}-16\,t^{5}-126\,t^{4}+864\,t^{3}-
1134\,t^{2}-1296\,t+4374 ) t}{t^{8}+36\,t^{6}+486\,t^{4}
+2916\,t^{2}+6561})$. Multiple points are indicated by red squares and isolated points by red stars.
      }}}
    \label{fig:ex-12}
  \end{center}
\end{figure}

\begin{figure}[t!]
\vspace{-0.5cm}
  \begin{center}
    \includegraphics[scale=0.29]{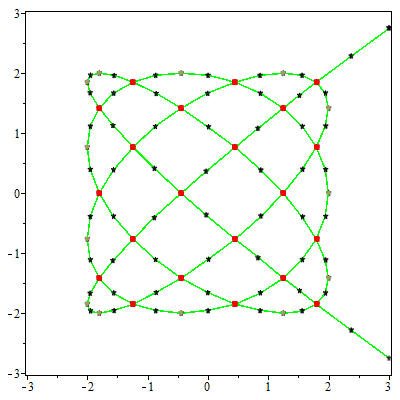}
    \includegraphics[scale=0.29]{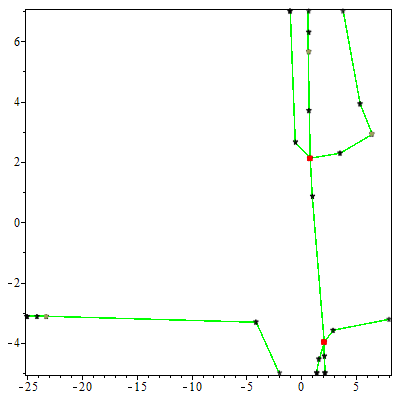}
    \vspace{-.2cm}
        \caption{\emph{\textmd{The left figure is the output of \textsc{ptopo} for the  parametric curve
        $({t}^{8}-8\,{t}^{6}+20\,{t}^{4}-16\,{t}^{2}+2,{t}^{7}-7\,{t}^{5}+14\,{
t}^{3}-7\,t)$
        while the right figure is the output for the curve
        $(\frac {37\,{t}^{3}-23\,{t}^{2}+87\,t+44}{29\,{t}^{3}+98\,{t}^{2}-23
\,t+10},\frac {-61\,{t}^{3}-8\,{t}^{2}-29\,t+95}{11\,{t}^{3}-49\,{t}
^{2}-47\,t+40})
$. Multiple points are indicated by red squares. \vspace{-1cm}
      }}}.
    \label{fig:ex-34}
  \end{center}
\end{figure}

\begin{theorem}[\ptopo and an abstract graph]
  \label{thm:ptopo-complexity-noBB}
  Consider a proper parametrization $\phi$ of curve $\cC$ involving
  polynomials of degree $d$ and bitsize $\tau$, as in Eq.~(\ref{param}),  that has no singularities at infinity.
  Alg.~\ref{alg:topol} outputs a graph $G$ that, if we embed it in
  $\RR^n$, then it is isotopic to $\cC$.
  It has worst case complexity
  \[
    \sOB(nd^6+nd^5\tau+d^4(n^2+n\tau)+d^3(n^2\tau+ n^3)+n^3d^2\tau),
  \]
  while its expected complexity is
  \[
	\sOB(d^6+d^5(n+\tau)+d^4(n^2+n\tau)+d^3(n^2\tau + n^3)+n^3d^2\tau),
  \]
  If $n = \OO(1)$, then  bounds become $\sOB(N^6)$, where
  $N = \max\{d, \tau\}$.
\end{theorem}

\begin{remark}
  If we are given a box $\mathcal{B} \subset \RR^n$ at the input,
  we slightly modify \ptopo, as follows:
  We
  discard the parameter values in
  $\mathtt{T}_C \cup \mathtt{T}_M \cup \mathtt{T}_E \cup \mathtt{T}_I$
  that correspond~to~points~not~contained~in~$\mathcal{B}$. The set of $G$'s
  vertices is constructed similarly.  To connect the vertices,
  we follow the same method with a minor modification: For~any~consecutive elements $t_1 < t_2$ in a list $\mathtt{T}_i$~with~more~than~two~elements, such that $t_1 \in \lambda(v_{i,1})$ and
  $t_2 \in \lambda(v_{i,2})$, we add the edge $\{v_{i,1}, v_{i,2}\}$
  if and only if $\phi(t_1)$, $\phi(t_2)$ are not both on the boundary
  of $\mathcal{B}$; or in other words $t_1$ and $t_2$ are not
  both in $\mathtt{T}_B$.
\end{remark}

\section{Isotopic embedding for the special cases of plane and space curves}
\label{Sembed}

In this section we elaborate on the isotopic embedding of the output graph $G$
of Thm.~\ref{thm:ptopo-complexity-noBB} for the case of plane and space parametric curves $\cC$.
We embed every edge of the abstract graph $G$ in the corresponding parametric
arc by sampling many parameter values in the associated parametric interval and
then connecting the corresponding points accordingly, in $\RR^2$ or $\RR^3$. The
larger the number of sampled parameters, the more likely it is for the
embedding to be isotopic to $\cC$. However, we might need a prohibitive large
number of points to sample; their number is related to the distance
between two branches of the curve.
We show that by introducing a few
additional points, we can replace the parametric arcs of the embedded graph with
straight line segments and count on it being isotopic to $\cC$.
Following \cite{ALCAZAR2020101830} closely,
if $X, Y \subset \RR^n$ are one-dimensional, then being isotopic implies that \textit{one of them can be
  deformed into the other without removing or introducing self-intersections}.



For plane curves, there is no need to take intermediate points on each parametric arc. We consider the embedding of the abstract graph $G$ in $\RR^2$ as a straight-line graph $\tilde{G}$, i.e., with straight lines for edges,
 whose vertices
are mapped to the corresponding points of the curve. The vertices of $\tilde{G}$ are all the singular and extreme points with respect to the $x$- and $y$- directions. Therefore, the edges of $\tilde{G}$ correspond to smooth and monotonous parametric arcs and so they cannot intersect but at their endpoints. The embedding $\tilde{G}$ is then trivially continuously deformed to $\mathcal{C}$. The above discussion summarizes as follows:

\begin{corollary}[\ptopo and isotopic embedding for plane curves]
\label{cor:embed2}
 Consider a proper parametrization $\phi$ of a curve $\cC$ in $\RR^2$ involving
  polynomials of degree $d$ and bitsize $\tau$, as in Eq.~(\ref{param}),  that has no singularities at infinity.
Alg.~\ref{alg:topol} computes an abstract graph whose straight-line embedding in $\RR^2$ is isotopic to $\cC$ in worst case complexity $\sOB(d^6+d^5\tau)$.
\end{corollary}

For  space curves, a straight-line embedding of $G$ is not guaranteed to be isotopic to $\cC$ for knots may be present.
To overcome this issue, we need to segment some  edges of $G$  into two or more edges. To find the extra vertices that we need to add to the graph, we follow a common approach \citep{ALCAZAR2010483,ALCAZAR2020101830,ELKAHOUI2008235,10.1145/1390768.1390778,CHENG201318} that projects the space curve to a plane one.
For a projection defined by the  map $\pi: \cC \rightarrow \RR^2$,
we write  $\tilde{\cC} = \pi(\cC)$.
We will ensure in the sequel that the following two conditions are satisfied:
\begin{enumerate}
\item[(C1)] $\cC$ has no asymptotes parallel to the direction of the projection.
\item[(C2)] The map $\pi$ is birational \cite[Def.~2.37]{SenWinPer-book-08}.
\end{enumerate}

The first condition is to ensure that the point at infinity $\mathbf{p}_\infty$ of $\cC$ exists if and only if the point at infinity of $\tilde{\cC}$ exists and is  equal to $\pi(\mathbf{p}_\infty)$ \cite[Lem.~10]{ALCAZAR2010483}; see Fig.\ref{fig:apparentB} for an instance where this condition is violated.
The second condition ensures that only a finite number of points on $\tilde{\cC}$  have more than one point as a preimage. We call these  points   \textit{apparent singularities} \citep{10.1145/1390768.1390778}; see Fig.~\ref{fig:apparentA}.
Thus, with this condition we avoid the "bad" cases where  two branches of $\cC$  project  to the same branch of $\tilde{\cC}$.

\begin{figure}
\label{fig:apparent}
\begin{center}
\begin{subfigure}[t]{0.4\textwidth}
\includegraphics[scale =0.6]{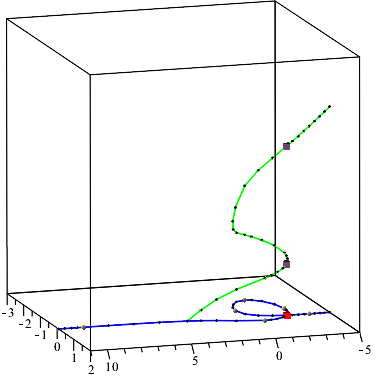}
\caption{\label{fig:apparentA}}
\end{subfigure}
\quad
\begin{subfigure}[t]{0.4\textwidth}
\includegraphics[scale =0.6]{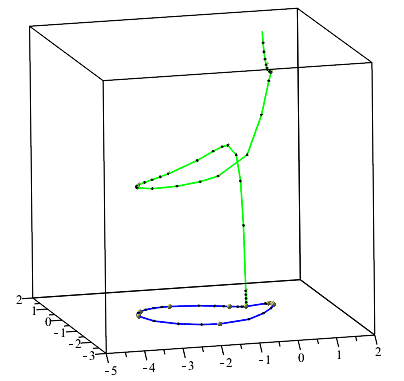}
\caption{\label{fig:apparentB}}
\end{subfigure}
\caption{\emph{(a) Graph of the curve parametrized by $\phi(t)=\big(\frac{7t^4-22t^3+55t^2+94t-87}{56t^4+62t^2-97t+73}, \frac{88t^5+4t^4-83t^3+10t^2-62t-82}{56t^4+62t^2-97t+73}, -\frac{95t^5-4t^4+83t^3-10t^2+62t+82}{56t^4+73}\big)$ (green) and its orthogonal projection on the $xy$-plane (blue). In red, a double point in the projected curve with two points in its preimage, i.e., apparent singularities (purple).
(b) Graph of the curve parametrized by $\phi(t)=\big( \frac{-7t^4+22t^3-55t^2-94t+87}{-56t^4-62t^2+97t-73}, \frac{-4t^4+83t^3-10t^2+62t+82}{-56t^4-62t^2+97t-73}, \frac{t^7-4t^4+83t^3-10t^2+62t+82}{-56t^4-73} \big)$ (green) and its orthogonal projection on the $xy$-plane (blue). The space curve does not have a point at infinity, whereas the plane curve has.  }}
\end{center}
\end{figure}

\begin{lemma}
\label{lem:projection}
 Consider a proper parametrization $\phi$ of curve $\cC$ in $\RR^3$ involving
  polynomials of degree $d$ and bitsize $\tau$, as in Eq.~(\ref{param}),  that has no singularities at infinity.
We compute a map $\pi : \cC \rightarrow \RR^2$ satisfying conditions (C1), (C2) in worst case complexity $\sOB(d^5+d^4\tau)$ and using a Las Vegas algorithm in expected complexity $\sOB(d^3+d^2\tau)$.
\end{lemma}

\begin{proof}
By \cite[Thm.~6.5, pg.~146]{Walker}, any space curve can be birationally projected to a plane curve.
We choose an integer $a$ uniformly at random from the set $\{1, \dots Kd^2\}$, where $K = \OO(1)$; we explain later in the proof  about the size of this set.
We define the mapping:
\begin{align}
	\pi :\quad & \CC^3 \rightarrow \CC^2 \nonumber\\
	 &(x,y,z) \mapsto (x, y+a z)
\end{align}

It will be useful in the sequel to regard the application of $\pi(\cdot)$ to $\phi(t)$ as being performed in two steps:
\begin{enumerate}
  \item Change of orthogonal basis of $\RR^3$: Let \(\{ \begin{pmatrix} 1 \\ 0\\ 0 \end{pmatrix}, \begin{pmatrix} 0 \\ 1\\-a \end{pmatrix}, \begin{pmatrix} 0 \\ a\\ 1 \end{pmatrix} \}\) be an orthogonal basis of $\RR^3$ and 
  $$
A=	\begin{pmatrix}
1 & 0 & 0\\
0 & 1 & a \\
0 & -a & 1\\
\end{pmatrix}.$$
In the new basis the curve is parametrized by:
$$(
\phi_1(t),
\phi_2(t) ,
\phi_3(t)
)
\cdot
	A
= 
(
\phi_1(t),
\phi_2(t) + a \phi_3(t) ,
-a\phi_2(t)+\phi_3(t)
)
$$
\item Orthogonal projection onto the first two coordinates: This yields the plane curve parametrized by $(
\phi_1(t),
\phi_2(t) + a \phi_3(t))$.
\end{enumerate}


For a given choice of $a$, we check if conditions (C1) and (C2) are satisfied:

For (C1): The direction of the projection is defined by the vector $(0,a,1)$.
So, $\cC$ has no asymptotes parallel to $(0,a,1)$ if and only if the curve with parametrization $\phi(t)\cdot A$
has no asymptotes parallel to $(0,a,1)\cdot A^{-1} = (0,0,1)$. We can check if this is the case by employing \cite[Lem.~9]{ALCAZAR2010483};
this is no more costly that solving a univariate polynomial of size $(\OO(d), \sOO(\tau))$, which costs $\sOB(d^3+d^2\tau)$ worst case.

Morover, there are $\OO(d)$ bad values of $a$ for whom (C1) does not hold: for any asymptote of $\cC$, there is a unique value of $a$ that maps it to an asymptote parallel to $(0,0,1)$ in the new basis. The asymptotes of $\cC$ are $\OO(d)$ since they occur at the poles of $\phi(t)$ and at the branches that extend to infinity. 

For (C2): Since $\phi(t)$ is proper, $\pi(\cdot)$ is birational if and only if $\pi(\phi(t))$ is also proper. We check the properness of this parametrization of $\tilde{\cC}$ in  $\sOB(d^3+d^2\tau)$ expected time, using Lem.~\ref{lem:check-proper}.

To find the values of $a$ that result in a `bad' map:
Let $\tilde{h_1}(s,t), \tilde{h_2}(s,t)$ be the polynomials of Eq.~(\ref{param}) associated to $\pi(\phi(t))$. The parametrization $\pi(\phi(t))$ is  proper if and only if $\gcd(\tilde{h_1}(s,t), \tilde{h_2}(s,t))=1$. If $\gcd(\tilde{h_1}(s,t), \tilde{h_2}(s,t)) \neq1$ then, by letting $R(s)=\res_t(\tilde{h_1}(s,t), \tilde{h_2}(s,t)) $, we have that $R(s)=0$. 
Notice that $R(s)$ is not always identically zero (e.g., for $\tilde{h_1}(s,t)=t+s,\tilde{h_2}(s,t)= t+s-1$ we get $R(s)=1$). We consider $R(s)$ as a polynomial in $\ZZ(a)[s]$:
\[
	R(s)= c_{d^2}(a)s^{d^2} +\dots +  c_{1}(a)s+ c_{0}(a) ,
\]
where $c_i \in \ZZ[a]$ is of size $(d,\sOO(d\tau))$ for $0 \le i \le d^2$ . The bad values of $a\in \RR$ satisfy then the equation:
 \[
 	 c_{d^2}^2(a)+\dots +  c_{1}^2(a)+ c_{0}^2(a) =0.
 \]
 The polynomial has degree $\OO(d^2)$ and so there are  $\OO(d^2)$ bad values to avoid. This points to the worst case complexity $\sOB(d^5+d^4\tau)$. 
 \end{proof}

Given a map $\pi$ computed through Lem.~\ref{lem:projection}, we find the
parameters that give the real multiple points of $\tilde{\cC}$ and its extreme
points with respect to the two coordinate axes. We add the corresponding
vertices to $G$ and we obtain an augmented graph, say $G'$. The straight-line
embedding of $G'$ in $\RR^2$ is isotopic to $\tilde{\cC}$ by
Cor.~\ref{cor:embed2}, possibly up to the isolated points \cite[Thm.~13 and
Lem.~15]{ALCAZAR2010483}. Then, by lifting this embedding to the corresponding
straight-line graph in $\RR^3$, we obtain a graph isotopic to $\cC$
\cite[Thm.~13 and Thm.~14]{ALCAZAR2010483}. The following theorem summarizes the
previous discussion and states its complexity:





\begin{theorem}[\ptopo and isotopic embedding for 3D space curves]
  \label{thm:ptopo-complexity-embedding3}
  Consider a proper parametrization $\phi$ of curve $\cC$ in $\RR^3$ involving
  polynomials of degree $d$ and bitsize $\tau$, as in Eq.~(\ref{param}), that has no singularities at infinity.
There is an algorithm that computes an abstract graph whose straight-line embedding in $\RR^3$ is isotopic to $\cC$  in worst case complexity $\sOB(d^6+d^5\tau)$.

\end{theorem}

\begin{proof}
Given a projection map $\pi(\cdot)$ such that  (C1) and (C2) hold, correctness follows from the previous discussion.
From Lem.~\ref{lem:projection}, we find such a map in expected complexity $\sOB(d^3+d^2\tau)$ and $\sOB(d^5+d^4\tau)$.
Using Alg.~\ref{alg:compute-points} for $\pi(\phi(t))$ we find the parameters of the extreme (in the coordinate directions) and real multiple points of $\tilde{\cC}$, say $\tilde{\mathtt{T}}$, in $\sOB(d^6+d^5\tau)$ (Lem.~\ref{lem:compl_points}).
Then, we employ Alg.~\ref{alg:topol} for $\phi(t)$, by augmenting the list of parameters that are treated with $\tilde{\mathtt{T}}$.
The last step dominates the complexity and is not affected by the addition of extra parameters to the list since $|\tilde{\mathtt{T}}|= \OO(d^2)$. The complexity result in Thm.~\ref{thm:ptopo-complexity-noBB} allows us to conclude.
\end{proof}

\begin{remark}
The complexity in Thm.~\ref{thm:ptopo-complexity-embedding3} is no higher than the complexity in Thm.~\ref{thm:ptopo-complexity-noBB}, i.e., no costlier than the computation of the topology itself.
\end{remark}



\begin{figure}

\begin{center}
\begin{subfigure}[t]{0.4\textwidth}
\begin{center}
  \includegraphics[scale=0.5]{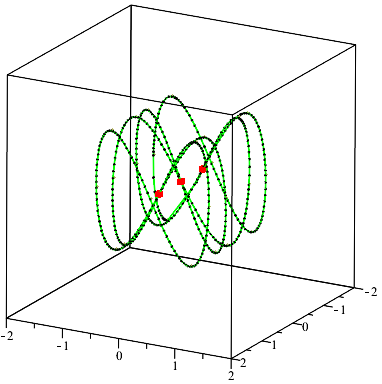}
    \caption{\label{fig:liscurve2}}
    \end{center}
\end{subfigure}
\quad
\begin{subfigure}[t]{0.4\textwidth}
\begin{center}
 \includegraphics[scale=0.5]{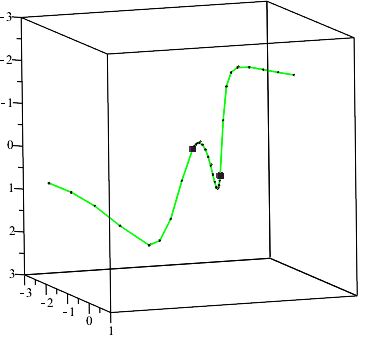}

     \caption{\label{fig:curve3d}}
     \end{center}
\end{subfigure}
\caption{\emph{Output of \ptopo for (a) a Lissajous curve parametrized by $\big( \frac{(-t^2-1)}{t^2+1}, \frac{-8t(t^2-1)(t^4-6t^2+1)}{t^8+4t^6+6t^4+4t^2+1}, \frac{-(16(t^4-6t^2+1))(t^2-1)t(t^8-28t^6+70t^4-28t^2+1)}{(t^16+8t^14+28t^12+56t^10+70t^8+56t^6+28t^4+8t^2+1)}\big)$. The multiple points are indicated by the red squares. (b) a space curve parametrized by $\big( \frac{-2t^2}{t^2+1}, \frac{-2t(3t^5-19t^4+8t^3-9t^2-t+2}{(t^2+1)(10t^4-6t^3+3t^2+1)}, t\big)$. The apparent singularities are indicated by the black squares.}}
\label{fig:ex3d1}
\end{center}
\end{figure}

\section{Multiplicities and characterization of singular points}
\label{Smult}





We say that a singularity is ordinary if it is at the intersection of smooth
branches only and the tangents to all branches are distinct \cite[p.54]{Walker}.
In all the other cases, we call it a \textit{non-ordinary} singularity.
The \textit{character} of a singular point is either ordinary or non-ordinary.
To determine the multiplicity and the character of each real singular point in $\CC$ we follow the method presented in the series of papers \cite{PEREZDIAZ2007835,Blasco2017ResultantsAS,BlaPer-rspc-19}.
They provide a complete characterization using resultant computations  that applies
to curves of any dimension. In the sequel, we present the basic ingredients of their approach
and we estimate the bit complexity of the algorithm.

Let $n=2$ and $H_i(s,t) = p_i(s)q_i(t)-p_i(t)q_i(s)$, for $i\in [2]$.
Consider  a point $\mathbf{p}\in \mathcal{C}$ given by the parameter values $\{s_1, \dots, s_k\}$, $k \ge 1$;
that is $\mathbf{p} = \phi(s_i)$, for all $i \in [k]$.
The \textit{fiber function} at $\mathbf{p}$ \cite[Def.~2]{BlaPer-rspc-19} is 
\[
	F_\mathbf{p}(t) = \gcd(H_1(s_j, t), H_2(s_j, t)) ,
\]
for any $j\in [k]$. It is a univariate polynomial, the real roots of which are the parameter values
that correspond to the point $\mathbf{p}$.
When the parametrization of the curve is proper, for any point $\mathbf{p}$ on $\mathcal{C}$ other than $\mathbf{p}_\infty$ it holds that $\deg(F_\mathbf{p}(t))= \text{mult}_\mathbf{p}(\mathcal{C})$ \cite[Cor.~1]{BlaPer-rspc-19}.
Also, when the parameters in $\{s_1, \dots, s_k\}$ are all real, $k$ equals the \textit{number of real branches} that go through $\mathbf{p}$.

To classify ordinary and non-ordinary singularities we proceed as follows:
For a point $\mathbf{p}\in \mathcal{C}$ the \textit{delta invariant}
$\delta_\mathbf{p}$ is a nonnegative integer that measures the number of double
points concentrated around $\mathbf{p}$. We can compute it by taking into account
the multiplicities of $\mathbf{p}$ and  its neighboring
singularities. \cite{BlaPer-rspc-19}  consider three different types of
non-ordinary singularities and use the delta invariant to distinguish them.
In particular:
\begin{enumerate}
  \item If $k = \text{mult}_\mathbf{p}(\cC)$ and
  $2 \delta_{\mathbf{p}} = \text{mult}_\mathbf{p}(\cC)  (\text{mult}_\mathbf{p}(\cC)  -1) $,
  then $\mathbf{p}$ is an ordinary singularity.
  \item If $k < \text{mult}_\mathbf{p}(\cC)$ and
  $2 \delta_{\mathbf{p}} =\text{mult}_\mathbf{p}(\cC)  (\text{mult}_\mathbf{p}(\cC)  -1)$,
  then $\mathbf{p}$ is a type $I$ non-ordinary singularity.
  \item If $k = \text{mult}_\mathbf{p}(\cC)$ and
  $2 \delta_{\mathbf{p}} >\text{mult}_\mathbf{p}(\cC)  (\text{mult}_\mathbf{p}(\cC)  -1)$,
  then $\mathbf{p}$ is a type $II$ non-ordinary singularity.
  \item If $k < \text{mult}_\mathbf{p}(\cC)$ and
  $2 \delta_{\mathbf{p}} > \text{mult}_\mathbf{p}(\cC)  (\text{mult}_\mathbf{p}(\cC)  -1)$,
  then $\mathbf{p}$ is a type $III$ non-ordinary singularity.
\end{enumerate}

\cite{BlaPer-rspc-19}
compute the delta invariant, $\delta_\mathbf{p}$, using the formula
\begin{equation}
\label{eq:delta1}
	\delta_{\mathbf{p}} = \frac{1}{2}\sum_{j=1}^k \sum_{t\text{ s.t. } h_1(s_j,t)=h_2(s_j,t)=0} \text{Int}_{h_1,h_2}(s_j,t), 
\end{equation}
where $\text{Int}_{h_1,h_2}(\alpha, \beta)$ is the \textit{intersection
  multiplicity} of the coprime polynomials $h_1(s,t)$ and $h_2(s,t)$ at a point
$(\alpha,\beta)$. Using a  well known result,  e.g., \cite[1.6]{FulInter} as
stated in \cite[Prop.~5]{BuseIntersection}, we can compute the intersection multiplicities
using resultant computations. Let $R(s)=\res_t(h_1(s,t),h_2(s,t))$.
For a root $\alpha$ of $R(s)$, its multiplicity $\mu(a)$ is equal to the sum of
the intersection multiplicities of solutions of the system of $\{h_1,h_2\}$ in
the form $(\alpha, t)$, and so
\begin{equation}
\label{eq:intersection}
	\mu(\alpha) =  \sum_{t\text{ s.t. } h_1(\alpha,t)=h_2(\alpha,t)=0} \text{Int}_{h_1,h_2}(\alpha,t) .
\end{equation}
Therefore, from Eq.~(\ref{eq:delta1}) and Eq.~(\ref{eq:intersection}), we conclude that:
\begin{equation}
\label{eq:delta2}
	\delta_{\mathbf{p}} = \frac{1}{2}\sum_{j=1}^k \mu(s_j) .
\end{equation}




For space curves, we birationally project $\cC$ to a plane curve \cite[Thm.~6.5, pg.~146]{Walker}. 
The multiplicities of the singular points of $\cC$ and their delta invariant are preserved under the birational map realizing the projection \cite[Prop.~1 and Cor.~4]{BlaPer-rspc-19}.
The pseudo code of the algorithm appears in  Alg.~\ref{alg:multiplicities}.

  \begin{algorithm2e}[!h]
   \SetKw{RET}{{\sc return}}
   \KwIn{Proper parametrization
    $\phi \in \ZZ(t)^n$ without singularity at infinity, as in Eq.~(\ref{param})}

   \KwOut{Multiplicities and characterization of points}

   \BlankLine












 %
 
  $\mathcal{S} \gets$ \FuncSty{Special\_Points}($\phi$)

 \If{n=2}{
   Compute polynomials $H_1(s,t), H_2(s,t)$ for $\phi$
 }
 \Else{
 	\Repeat{$\tilde{\phi}(t) $ is proper}{
   Choose integers $a_3, \dots, a_n$ at random from $\{1,\dots, Kd^n\}$, where $K = \OO(1)$.

   $\tilde{\phi}(t) \gets (\phi_1(t), \phi_2(t) + a_3 \phi_3(t) + \dots+ a_n\phi_n(t))$}

  Compute polynomials $H_1(s,t), H_2(s,t)$ for $\tilde{\phi}$
 }

 \For{$\mathbf{p} \in \mathcal{S}$}{

	$ M_\mathbf{p} \gets \left\{ s\in \RR : \phi(s)=\mathbf{p}\right\}$ \tcp{parameters that give the same point $\mathbf{p}$}

	$k \gets |M_\mathbf{p}|$ 	\tcp{number of real branches that go through $\mathbf{p}$}

	Take $s_0 \in M_\mathbf{p}$


	$\text{mult}_\mathbf{p}(\mathcal{C}) \gets \deg(\gcd(H_1(s_0, t),H_2(s_0, t)))$ 	\tcp{multiplicity}

	\tcp{compute delta invariant}

	$\delta_\mathbf{p} \gets 0$

	\For{$s_j\in  M_\mathbf{p}$}{

		$\mu(s_j) \gets$ multiplicity of $s_j$ as a root of $\res_t(H_1(s,t)/(s-t), H_2(s,t)/(s-t))$ 

		$\delta_\mathbf{p} \gets \delta_\mathbf{p} + \mu(s_j)$

	}
	$\delta_\mathbf{p} \gets \delta_\mathbf{p} / 2$

	    \If{$k = \text{mult}_\mathbf{p}(\mathcal{C})$ and $2\delta_\mathbf{p} =\text{mult}_\mathbf{p}(\mathcal{C})(\text{mult}_\mathbf{p}(\mathcal{C})-1)$ }{
                 \Return  $\mathbf{p}$ is an ordinary singularity
         }
          \ElseIf{$k < \text{mult}_\mathbf{p}(\mathcal{C})$ and $2\delta_\mathbf{p} =  \text{mult}_\mathbf{p}(\mathcal{C})( \text{mult}_\mathbf{p}(\mathcal{C})-1)$}{
                  \Return  $\mathbf{p}$ is a type $I$ non-ordinary singularity
         }
          \ElseIf{$k = \text{mult}_\mathbf{p}(\mathcal{C})$ and $2\delta_\mathbf{p} > \text{mult}_\mathbf{p}(\mathcal{C})( \text{mult}_\mathbf{p}(\mathcal{C})-1)$}{
          	\Return  $\mathbf{p}$ is a type $II$ non-ordinary singularity
          }
          \ElseIf{$k < \text{mult}_\mathbf{p}(\mathcal{C})$ and $2\delta_\mathbf{p} >  \text{mult}_\mathbf{p}(\mathcal{C})( \text{mult}_\mathbf{p}(\mathcal{C})-1)$}{
                 \Return  $\mathbf{p}$ is a type $III$ non-ordinary singularity

          }

}

    \caption{\FuncSty{CharacterizeSingularPoints}($\phi, \mathcal{S}$)}
    \label{alg:multiplicities}
  \end{algorithm2e}

\begin{theorem}
\label{thm:multiplicities}
Let $\mathcal{C}$ be a curve with a proper parametrization $\phi(t)$ as in Eq.~(\ref{param}), that has no singularities at infinity. Alg.~\ref{alg:multiplicities} computes the singular points of $\cC$, their multiplicity and character (ordinary/non-ordinary) in
\[
	\sOB(d^6+d^5\tau)
\]
worst-case complexity when $n=2$ and for $n>2$ in expected complexity
\[
  \sOB(d^6+d^5(n+\tau)+d^4(n^2+n\tau)+d^3(n^2\tau + n^3)+n^3d^2\tau).
\]

\end{theorem}

\begin{proof}
We compute the parameters that give the singular points of $\cC$ using Alg.~\ref{alg:compute-points} in $ \sOB(d^6+d^5(n+\tau)+d^4(n^2+n\tau)+d^3(n^2\tau + n^3)+n^3d^2\tau)$ when $n>2$, which becomes worst-case when $n=2$ (Lem.~\ref{lem:compl_points}).

When $n>2$, lines 5-8 compute a birational projection of $\cC$ to a plane curve parametrized by $\tilde{\phi}(t)$ (note that the projection is birational if and only if $\tilde{\phi}(t)$ is proper since $\phi(t)$ is proper).
The expected complexity of this is $\sOB(d^3+nd^2\tau)$ by slightly adapting the proof of  Lem.~\ref{lem:projection}.

 We can group the parameter values that give singular points using Lem.~\ref{lem:group-params} in $\OO(d^2+nd)$ arithmetic operations.
We have that $\gcd(H_1(s,t), H_2(s,t))= s-t$.
For $h_1(s,t) = H_1(s,t)/(s-t)$, $h_2(s,t) = H_2(s,t)/(s-t)$, we compute a triangular decomposition of the system $\{h_1(s,t)=h_2(s,t)=0\}$ which consists of the systems $\{(A_i(s), B_i(s,t))\}_{i\in \mathcal{I}}$.
For any root $\alpha$ of $A_i$, $B_i(\alpha, t)$ is of degree $i$ and equals $\gcd(h_1(\alpha, t), h_2(\alpha, t))$ (up to a constant
factor).
By \cite[Prop.~16]{Bouz15}, the triangular decomposition is computed in $\sOB(d^6+d^5\tau)$ worst case\footnote{In a Las Vegas setting this computation could be reduced to $\sOB(d^4+d^3\tau)$, but since it does not affect the total complexity we chose not to expand onto this.}.

 For a singular point $\mathbf{p}$ and $s_0 \in M_\mathbf{p}$:
To compute the degree of $\gcd(H_1(s_0), H_2(s_0))$, since $s_0$ is a root of $\res_t(h_1(s,t), h_2(s,t))$, it suffices to find for which $i\in \mathcal{I}$ $A_i(s_0)=0$. The latter is not immediate when $s_0$ is given in isolated interval representation as a root of  $\res_t(h_1(s,t), h_2(s,t))$. However, we can isolate the roots of all $A_i$ by taking care that the isolating invervals are small enough so that the isolating interval of $s_0$ intersects only one of them. Because they are roots of the same polynomial this cannot exceed the complexity of isolating the roots of $\res_t(h_1(s,t), h_2(s,t))$.
This is can be done in $\sOB(d^6+d^5\tau)$.
In the same bit complexity, we have the multiplicity $\mu(s_j)$ of $s_j$ as a root of $\res_t(h_1,h_2)$.
  \end{proof}





\section{Implementation and examples}

\label{sec:implementation}

\begin{figure}

\begin{center}
\begin{subfigure}[t]{0.4\textwidth}
\includegraphics[scale =0.6]{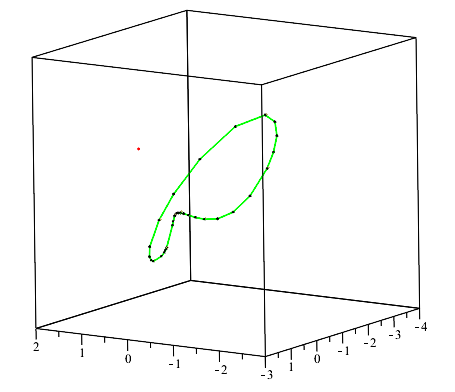}
\caption{\label{fig:ex3d1}}
\end{subfigure}
\quad
\begin{subfigure}[t]{0.4\textwidth}
\includegraphics[scale =0.49]{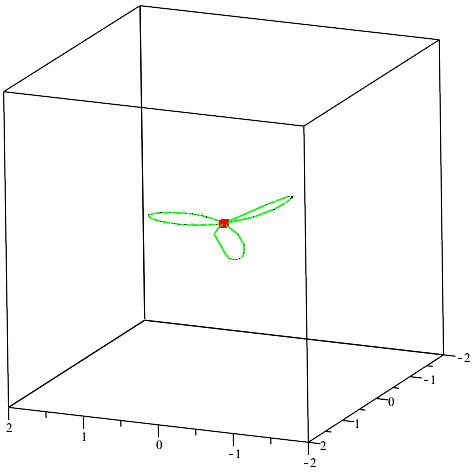}
\caption{\label{fig:ex3d2}}
\end{subfigure}

\caption{\emph{Output of \ptopo for a curve parametrized by (a) $\big(\frac{-7t^4+22t^3-55t^2-94t+87}{-56t^4-62t^2+97t-73}, \frac{-4t^4*83t^3-10t^2-62t-82}{-56t^4-62t^2+97t-73}, \frac{-4t^4*83t^3-10t^2-62t-82}{-56t^4-62t^2+97t-73} \big)$
(b) $\big( \frac{-3t^2+1}{(t^2+1)^2}, \frac{(-3t^2+1)t}{(t^2+1)^2}, \frac{(-3t^2+1)t^3}{(t^2+1)^4\big)}$. The red star in (a) corresponds to an isolated point, whereas the red square in (b) corresponds to a multiple point. Examples are taken from \cite{ALCAZAR2010483}.}}

\label{fig:ex3d2}
\end{center}
\end{figure}

\ptopo is implemented in
\maple\footnote{\url{https://gitlab.inria.fr/ckatsama/ptopo}}
 for plane and 3D curves.
A typical output appears in Figures~\ref{fig:ex-12}, \ref{fig:ex-34}, \ref{fig:ex3d1},
and \ref{fig:ex3d2}. Besides the visualization the software computes all the points of interest of curve (singular, extreme, etc) in isolating interval representation
as well as in suitable floating point approximations.

We build upon the real root isolation routines of \maple's
\texttt{RootFinding} library and the
\textsc{slv}
package \citep{det-jsc-2009},
to use a certified implementation of general purpose exact computations
with one and two real algebraic numbers, like comparison and sign evaluations,
as well as exact (bivariate) polynomial solving.

\ptopo computes the topology and visualizes parametric curves in two
(and in the near future in three) dimensions.
For a given parametric representation of a curve, \ptopo computes
the special points on the curve, the characteristic box, the
corresponding graph, and then it visualizes the curve (inside the box).
The computation, in all examples from literature we tested,
takes less than a second in a MacBook laptop, running \maple 2020.
We refer the reader
to the website of the software and to \citep{krtz-ptopo-soft} for further details.

\freefootnote{{
\begin{wrapfigure}{l}{0.05\textwidth}
\begin{center}
\vspace{-0.85cm}
\includegraphics[scale=0.04]{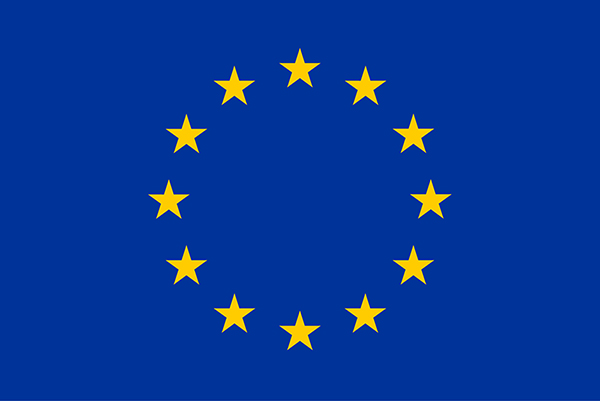}
\end{center}
\end{wrapfigure}
{\footnotesize This project has received funding from the European Union's Horizon 2020 research and innovation program under the
Marie Sk\l odowska-Curie grant agreement No 754362.}
}}

\paragraph*{Acknowledgments}
We acknowledge partial support by the
Fondation Math\'ematique Jacques Hadamard through the PGMO grand ALMA,
by ANR JCJC GALOP (ANR-17-CE40-0009), by the  PHC GRAPE,
by the projects 118F321 under the program 2509, 118C240 under the program 2232,
and 117F100 under the program 3501 of the Scientific and Technological Research Council of Turkey.

\bibliographystyle{elsarticle-harv}
\bibliography{param}

\begin{thebibliography}{59}
\expandafter\ifx\csname natexlab\endcsname\relax\def\natexlab#1{#1}\fi
\expandafter\ifx\csname url\endcsname\relax
  \def\url#1{\texttt{#1}}\fi
\expandafter\ifx\csname urlprefix\endcsname\relax\def\urlprefix{URL }\fi

\bibitem[{Abhyankar and Bajaj(1989)}]{10.1145/77269.77273}
Abhyankar, S.~S., Bajaj, C.~J., 1989. Automatic parameterization of rational
  curves and surfaces {IV}: Algebraic space curves. ACM Trans. Graph. 8~(4),
  325--334.
\newline\urlprefix\url{https://doi.org/10.1145/77269.77273}

\bibitem[{Alberti et~al.(2008)Alberti, Mourrain, and Wintz}]{ALBERTI2008631}
Alberti, L., Mourrain, B., Wintz, J., 2008. Topology and arrangement
  computation of semi-algebraic planar curves. CAGD 25~(8), 631 -- 651.
\newline\urlprefix\url{http://www.sciencedirect.com/science/article/pii/S0167839608000460}

\bibitem[{Alc\'azar et~al.(2020)Alc\'azar, Caravantes, D\'iaz, and
  Tsigaridas}]{ALCAZAR2020101830}
Alc\'azar, J.~G., Caravantes, J., D\'iaz, G.~M., Tsigaridas, E., 2020.
  Computing the topology of a plane or space hyperelliptic curve. Computer
  Aided Geometric Design 78, 101830.
\newline\urlprefix\url{http://www.sciencedirect.com/science/article/pii/S0167839620300170}

\bibitem[{Alc\'azar and D\'iaz-Toca(2010)}]{ALCAZAR2010483}
Alc\'azar, J.~G., D\'iaz-Toca, G.~M., 2010. Topology of 2{D} and 3{D} rational
  curves. CAGD 27~(7), 483 -- 502.
\newline\urlprefix\url{http://www.sciencedirect.com/science/article/pii/S0167839610000737}

\bibitem[{Basu et~al.(2003)Basu, Pollack, and {M-F.}Roy}]{BPR03}
Basu, S., Pollack, R., {M-F.}Roy, 2003. Algorithms in Real Algebraic Geometry.
  Vol.~10 of Algorithms and Computation in Mathematics. Springer-Verlag.

\bibitem[{Bernardi et~al.(2016)Bernardi, Gimigliano, and Id\`a}]{bernardi}
Bernardi, A., Gimigliano, A., Id\`a, M., 09 2016. Singularities of plane
  rational curves via projections. J. Symb. Comput.

\bibitem[{Blasco and P{\'e}rez-D{\'i}az(2017)}]{Blasco2017ResultantsAS}
Blasco, A., P{\'e}rez-D{\'i}az, S., 2017. Resultants and singularities of
  parametric curves. arXiv.
\newline\urlprefix\url{https://arxiv.org/abs/1706.08430s}

\bibitem[{Blasco and P{\'e}rez-D{\'\i}az(2019)}]{BlaPer-rspc-19}
Blasco, A., P{\'e}rez-D{\'\i}az, S., 2019. An in depth analysis, via
  resultants, of the singularities of a parametric curve. CAGD 68, 22--47.

\bibitem[{Boissonnat and Teillaud(2006)}]{bt-ecgcs-06}
Boissonnat, J.-D., Teillaud, M. (Eds.), 2006. Effective Computational Geometry
  for Curves and Surfaces. Springer-Verlag, Mathematics and Visualization.

\bibitem[{Bouzidi et~al.(2015{\natexlab{a}})Bouzidi, Lazard, Moroz, Pouget,
  Rouillier, and Sagraloff}]{Bouz15}
Bouzidi, Y., Lazard, S., Moroz, G., Pouget, M., Rouillier, F., Sagraloff, M.,
  02 2015{\natexlab{a}}. Improved algorithms for solving bivariate systems via
  rational univariate representations.

\bibitem[{Bouzidi et~al.(2016)Bouzidi, Lazard, Moroz, Pouget, Rouillier, and
  Sagraloff}]{bouzidi:hal-01342211}
Bouzidi, Y., Lazard, S., Moroz, G., Pouget, M., Rouillier, F., Sagraloff, M.,
  2016. {Solving bivariate systems using Rational Univariate Representations}.
  {J. Complexity} 37, 34--75.
\newline\urlprefix\url{https://hal.inria.fr/hal-01342211}

\bibitem[{Bouzidi et~al.(2013)Bouzidi, Lazard, Pouget, and
  Rouillier}]{Bouzidi:2013:RUR:2465506.2465519}
Bouzidi, Y., Lazard, S., Pouget, M., Rouillier, F., 2013. Rational univariate
  representations of bivariate systems and applications. In: Proc. 38th Int'l
  Symp. on Symbolic and Algebraic Computation. ISSAC '13. ACM, NY, USA, pp.
  109--116.
\newline\urlprefix\url{http://doi.acm.org/10.1145/2465506.2465519}

\bibitem[{Bouzidi et~al.(2015{\natexlab{b}})Bouzidi, Lazard, Pouget, and
  Rouillier}]{blpr-jsc-2015}
Bouzidi, Y., Lazard, S., Pouget, M., Rouillier, F., 2015{\natexlab{b}}.
  Separating linear forms and rational univariate representations of bivariate
  systems. {J. Symb. Comput.}, 84--119.
\newline\urlprefix\url{https://hal.inria.fr/hal-00977671}

\bibitem[{Bus{\'e}(2014)}]{buse:hal-00847802}
Bus{\'e}, L., 2014. {Implicit matrix representations of rational B{\'e}zier
  curves and surfaces}. {Computer-Aided Design} 46, 14--24.
\newline\urlprefix\url{https://hal.inria.fr/hal-00847802}

\bibitem[{Bus{\'e} and D'Andrea(2012)}]{BusDAd-sing-plane-12}
Bus{\'e}, L., D'Andrea, C., 2012. Singular factors of rational plane curves. J.
  Algebra 357, 322--346.

\bibitem[{Bus{\'e} et~al.(2005)Bus{\'e}, Khalil, and
  Mourrain}]{BuseIntersection}
Bus{\'e}, L., Khalil, H., Mourrain, B., 2005. Resultant-based methods for plane
  curves intersection problems. In: Ganzha, V.~G., Mayr, E.~W., Vorozhtsov,
  E.~V. (Eds.), Computer Algebra in Scientific Computing. Springer Berlin
  Heidelberg, Berlin, Heidelberg, pp. 75--92.

\bibitem[{Bus{\'e} et~al.(2019)Bus{\'e}, Laroche, and
  Y{\i}ld{\i}r{\i}m}]{BLY0-spm-2019}
Bus{\'e}, L., Laroche, C., Y{\i}ld{\i}r{\i}m, F., 2019. Implicitizing rational
  curves by the method of moving quadrics. Computer-Aided Design 114, 101--111.

\bibitem[{Bus{\'e} and Luu~Ba(2010)}]{buse:inria-00468964}
Bus{\'e}, L., Luu~Ba, T., Oct. 2010. {Matrix-based Implicit Representations of
  Rational Algebraic Curves and Applications}. {CAGD} 27~(9), 681--699.
\newline\urlprefix\url{https://hal.inria.fr/inria-00468964}

\bibitem[{Caravantes et~al.(2014)Caravantes, Fioravanti, Gonzalez-Vega, and
  Necula}]{Caravantes2014ComputingTT}
Caravantes, J., Fioravanti, M., Gonzalez-Vega, L., Necula, I., 2014. Computing
  the topology of an arrangement of implicit and parametric curves given by
  values. In: Gerdt, V.~P., Koepf, W., Seiler, W.~M., Vorozhtsov, E.~V. (Eds.),
  Computer Algebra in Scientific Computing. Springer, Cham, pp. 59--73.

\bibitem[{Cheng et~al.(2013)Cheng, Jin, and Lazard}]{CHENG201318}
Cheng, J.-S., Jin, K., Lazard, D., 2013. Certified rational parametric
  approximation of real algebraic space curves with local generic position
  method. Journal of Symbolic Computation 58, 18 -- 40.
\newline\urlprefix\url{http://www.sciencedirect.com/science/article/pii/S0747717113000953}

\bibitem[{Chionh and Sederberg(2001)}]{CHIONH200121}
Chionh, E.-W., Sederberg, T.~W., 2001. On the minors of the implicitization
  b\'ezout matrix for a rational plane curve. CAGD 18~(1), 21 -- 36.

\bibitem[{Cox et~al.(2011)Cox, Kustin, Polini, and Ulrich}]{Cox-syz}
Cox, D., Kustin, A., Polini, C., Ulrich, B., 02 2011. A study of singularities
  on rational curves via syzygies. Memoirs of the American Mathematical Society
  222.

\bibitem[{Diatta et~al.(2018)Diatta, Diatta, Rouillier, Roy, and
  Sagraloff}]{DDRRS-topo-2018}
Diatta, D.~N., Diatta, S., Rouillier, F., Roy, M.-F., Sagraloff, M., 2018.
  Bounds for polynomials on algebraic numbers and application to curve
  topology. arXiv preprint arXiv:1807.10622 (To appear in Discrete \&
  Computational Geometry).

\bibitem[{Diatta et~al.(2008)Diatta, Mourrain, and
  Ruatta}]{10.1145/1390768.1390778}
Diatta, D.~N., Mourrain, B., Ruatta, O., 2008. On the computation of the
  topology of a non-reduced implicit space curve. In: Proceedings of the
  Twenty-First International Symposium on Symbolic and Algebraic Computation.
  ISSAC '08. Association for Computing Machinery, New York, NY, USA, p.
  47–54.
\newline\urlprefix\url{https://doi.org/10.1145/1390768.1390778}

\bibitem[{Diochnos et~al.(2009)Diochnos, Emiris, and Tsigaridas}]{det-jsc-2009}
Diochnos, D.~I., Emiris, I.~Z., Tsigaridas, E.~P., 2009. On the asymptotic and
  practical complexity of solving bivariate systems over the reals. J. Symb.
  Comput. 44~(7), 818--835, (Special issue on ISSAC 2007).

\bibitem[{Farouki et~al.(2010)Farouki, Giannelli, and
  Sestini}]{10.1007/978-3-642-11620-9_13}
Farouki, R.~T., Giannelli, C., Sestini, A., 2010. Geometric design using space
  curves with rational rotation-minimizing frames. In: D{\ae}hlen, M., Floater,
  M., Lyche, T., Merrien, J.-L., M{\o}rken, K., Schumaker, L.~L. (Eds.),
  Mathematical Methods for Curves and Surfaces. Springer, pp. 194--208.

\bibitem[{Fulton(1969)}]{Fulton}
Fulton, W., 1969. Algebraic Curves. An Introduction to Algebraic Geometry.
  Addison~Wesley.

\bibitem[{Fulton(1984)}]{FulInter}
Fulton, W., 1984. Introduction to Intersection Theory in Algebraic Geometry.
  Cbms Regional Conference Series in Mathematics. American Mathematical
  Society.

\bibitem[{Gao and Chou(1992)}]{GAO1992459}
Gao, X.-S., Chou, S.-C., 1992. Implicitization of rational parametric
  equations. J. Symb. Comput. 14~(5), 459 -- 470.
\newline\urlprefix\url{http://www.sciencedirect.com/science/article/pii/074771719290017X}

\bibitem[{Gutierrez et~al.(2002{\natexlab{a}})Gutierrez, Rubio, and
  Sevilla}]{GUTIERREZ2002545}
Gutierrez, J., Rubio, R., Sevilla, D., 2002{\natexlab{a}}. On multivariate
  rational function decomposition. J. Symb. Comput. 33~(5), 545 -- 562.
\newline\urlprefix\url{http://www.sciencedirect.com/science/article/pii/S0747717100905297}

\bibitem[{Gutierrez et~al.(2002{\natexlab{b}})Gutierrez, Rubio, and Yu}]{DRes}
Gutierrez, J., Rubio, R., Yu, J.-T., 08 2002{\natexlab{b}}. D-resultant for
  rational functions. Proc. American Mathematical Society 130.

\bibitem[{Jia et~al.(2018)Jia, Shi, and Chen}]{JIA20182}
Jia, X., Shi, X., Chen, F., 2018. Survey on the theory and applications of
  $\mu$-bases for rational curves and surfaces. J. Comput. Appl. Math. 329,
  2--23.

\bibitem[{Kahoui(2008)}]{ELKAHOUI2008235}
Kahoui, M.~E., 2008. Topology of real algebraic space curves. J. Symb. Comput.
  43~(4), 235 -- 258.
\newline\urlprefix\url{http://www.sciencedirect.com/science/article/pii/S0747717107001265}

\bibitem[{Katsamaki et~al.(2020{\natexlab{a}})Katsamaki, Rouillier, Tsigaridas,
  and Zafeirakopoulos}]{krtz-ptopo-20}
Katsamaki, C., Rouillier, F., Tsigaridas, E., Zafeirakopoulos, Z.,
  2020{\natexlab{a}}. On the geometry and the topology of parametric curves.
  In: Proc. 45th Int'l Symposium on Symbolic and Algebraic Computation. ISSAC
  '20. ACM, NY, USA, p. 281–288.

\bibitem[{Katsamaki et~al.(2020{\natexlab{b}})Katsamaki, Rouillier, Tsigaridas,
  and Zafeirakopoulos}]{krtz-ptopo-soft}
Katsamaki, C., Rouillier, F., Tsigaridas, E.~P., Zafeirakopoulos, Z.,
  2020{\natexlab{b}}. {PTOPO: A M}aple package for the topology of parametric
  curves. {ACM} Commun. Comput. Algebra 54~(2), 49--52.

\bibitem[{Kobel and Sagraloff(2014)}]{Kobel-complexity}
Kobel, A., Sagraloff, M., 08 2014. On the complexity of computing with planar
  algebraic curves. J. Complexity 31.

\bibitem[{Li and Cripps(1997)}]{LI1997491}
Li, Y.-M., Cripps, R.~J., 1997. Identification of inflection points and cusps
  on rational curves. CAGD 14~(5), 491 -- 497.
\newline\urlprefix\url{http://www.sciencedirect.com/science/article/pii/S0167839696000416}

\bibitem[{Lickteig and Roy(2001)}]{LicRoy-sub-res-01}
Lickteig, T., Roy, M.-F., Mar. 2001. Sylvester{\textendash}{H}abicht sequences
  and fast {C}auchy index computation. J. Symb. Comput. 31~(3), 315--341.
\newline\urlprefix\url{https://doi.org/10.1006/jsco.2000.0427}

\bibitem[{Mahler(1962)}]{Mahler-ineq-mpoly-62}
Mahler, K., 1962. On some inequalities for polynomials in several variables. J.
  London Mathematical Society 1~(1), 341--344.

\bibitem[{Manocha and Canny(1992)}]{MANOCHA19921}
Manocha, D., Canny, J.~F., 1992. Detecting cusps and inflection points in
  curves. CAGD 9~(1), 1 -- 24.
\newline\urlprefix\url{http://www.sciencedirect.com/science/article/pii/016783969290050Y}

\bibitem[{Pan(2002)}]{Pan02jsc}
Pan, V., 2002. Univariate polynomials: Nearly optimal algorithms for numerical
  factorization and rootfinding. J. Symb. Comput. 33~(5), 701--733.

\bibitem[{Pan and Tsigaridas(2017)}]{PAN2017138}
Pan, V., Tsigaridas, E., 2017. Accelerated approximation of the complex roots
  and factors of a univariate polynomial. Theor. Computer Science 681, 138 --
  145.
\newline\urlprefix\url{http://www.sciencedirect.com/science/article/pii/S0304397517302529}

\bibitem[{Park(2002)}]{park02}
Park, H., 08 2002. Effective computation of singularities of parametric affine
  curves. J. Pure and Applied Algebra 173, 49--58.

\bibitem[{P{\'e}rez-D{\'\i}az(2006)}]{Perez-proper-06}
P{\'e}rez-D{\'\i}az, S., 2006. On the problem of proper reparametrization for
  rational curves and surfaces. CAGD 23~(4), 307--323.

\bibitem[{P\'erez-D\'iaz(2007)}]{PEREZDIAZ2007835}
P\'erez-D\'iaz, S., 2007. Computation of the singularities of parametric plane
  curves. J. Symb. Comput. 42~(8), 835 -- 857.
\newline\urlprefix\url{http://www.sciencedirect.com/science/article/pii/S0747717107000636}

\bibitem[{Recio(2007)}]{Carlo-2007}
Recio, C. A.~T., 02 2007. Plotting missing points and branches of real
  parametric curves. Applicable Algebra in Engineering, Communication and
  Computing 18.

\bibitem[{Rouillier(1999)}]{Rouillier1999SolvingZS}
Rouillier, F., 1999. Solving zero-dimensional systems through the rational
  univariate representation. Applicable Algebra in Engineering, Communication
  and Computing 9, 433--461.

\bibitem[{Rubio et~al.(2009)Rubio, Serradilla, and V\'elez}]{RUBIO2009490}
Rubio, R., Serradilla, J., V\'elez, M., 2009. Detecting real singularities of a
  space curve from a real rational parametrization. J. Symb. Comput. 44~(5),
  490 -- 498.
\newline\urlprefix\url{http://www.sciencedirect.com/science/article/pii/S0747717108001405}

\bibitem[{Sch{\"o}nhage(1988)}]{SCHONHAGE1988365}
Sch{\"o}nhage, A., 1988. Probabilistic computation of integer polynomial gcds.
  J. Algorithms 9~(3), 365 -- 371.
\newline\urlprefix\url{http://www.sciencedirect.com/science/article/pii/0196677488900272}

\bibitem[{Sederberg(1986)}]{Sed-improper-param-86}
Sederberg, T.~W., May 1986. Improperly parametrized rational curves. CAGD
  3~(1), 67--75.
\newline\urlprefix\url{http://www.sciencedirect.com/science/article/pii/0167839686900257}

\bibitem[{Sederberg and Chen(1995)}]{SedChen-implicit-95}
Sederberg, T.~W., Chen, F., 1995. Implicitization using moving curves and
  surfaces. In: Proc. of the 22nd Annual Conference on Computer Graphics and
  Interactive Techniques. SIGGRAPH '95. NY, USA, pp. 301--308.
\newline\urlprefix\url{https://doi.org/10.1145/218380.218460}

\bibitem[{Sendra(2002)}]{Sendra2002NormalPO}
Sendra, J.~R., 2002. Normal parametrizations of algebraic plane curves. J.
  Symb. Comput. 33, 863--885.

\bibitem[{Sendra and Winkler(1999)}]{Sendra:1999:ARR:2378083.2378093}
Sendra, J.~R., Winkler, F., Apr. 1999. Algorithms for rational real algebraic
  curves. Fundam. Inf. 39~(1,2), 211--228.
\newline\urlprefix\url{http://dl.acm.org/citation.cfm?id=2378083.2378093}

\bibitem[{Sendra et~al.(2008)Sendra, Winkler, and
  P{\'e}rez-D{\'\i}az}]{SenWinPer-book-08}
Sendra, J.~R., Winkler, F., P{\'e}rez-D{\'\i}az, S., 2008. Rational algebraic
  curves. Algorithms and Computation in Mathematics 22.

\bibitem[{Strzebonski and Tsigaridas(2019)}]{STRZEBONSKI201931}
Strzebonski, A., Tsigaridas, E., 2019. Univariate real root isolation in an
  extension field and applications. J. Symb. Comput. 92, 31 -- 51.
\newline\urlprefix\url{http://www.sciencedirect.com/science/article/pii/S0747717117301256}

\bibitem[{van~den Essen and YU(1997)}]{Essen}
van~den Essen, A., YU, J.-T., 01 1997. The {D}-resultant, singularities and the
  degree of unfaithfulness. Proc. American Mathematical Society 125.

\bibitem[{von~zur Gathen and Gerhard(2013)}]{gg-mca-13}
von~zur Gathen, J., Gerhard, J., 2013. Modern computer algebra, 3rd Edition.
  Cambridge University Press.

\bibitem[{Walker(1978)}]{Walker}
Walker, R.~J., 1978. Algebraic curves. Springer-Verlag.

\bibitem[{Yap(1999)}]{10.5555/328435}
Yap, C.~K., 1999. Fundamental Problems of Algorithmic Algebra. Oxford
  University Press, Inc., USA.

\end{thebibliography}

\end{document}